\documentclass[a4paper,USenglish,cleveref, autoref, thm-restate]{lipics-v2021}
\usepackage[T1]{fontenc}
\usepackage[utf8]{inputenc}
\usepackage{numprint}

\usepackage{wrapfig}
\usepackage{amsmath}
\usepackage{mathrsfs}
\usepackage{amssymb}
\usepackage{amsthm}
\usepackage{amsfonts}
\usepackage{wasysym}
\usepackage{tikz}
\usepackage[ruled, vlined, linesnumbered]{algorithm2e}
\usepackage{stmaryrd}
\usepackage{enumitem}
\usepackage{caption}
\usepackage{subcaption}
\usetikzlibrary{automata}
\usetikzlibrary{shapes}

\nolinenumbers

\newcommand{\<}{\langle}
\renewcommand{\>}{\rangle}

\newcommand{\bsigma}{\bar{\sigma}}
\newcommand{\btau}{\bar{\tau}}

\renewcommand{\|}{\upharpoonright}

\newcommand{\N}{\mathbb{N}}
\newcommand{\Z}{\mathbb{Z}}

\newcommand{\playcircle}{{\tiny{\Circle}}}

\renewcommand{\l}{\ell}
\renewcommand{\epsilon}{\varepsilon}
\renewcommand{\phi}{\varphi}

\newcommand{\card}{\mathsf{card}}

\newcommand{\Plays}{\mathsf{Plays}}
\newcommand{\Hist}{\mathsf{Hist}}

\newcommand{\nego}{\mathsf{nego}}
\newcommand{\lCons}{\lambda\mathsf{Cons}}

\newcommand{\lRat}{\lambda\mathsf{Rat}}

\newcommand{\Abs}{\mathsf{Abs}}
\newcommand{\Conc}{\mathsf{Conc}}

\newcommand{\Req}{\mathsf{Req}}
\newcommand{\Inf}{\mathsf{Inf}}
\newcommand{\InfTr}{\mathsf{InfTr}}
\newcommand{\Occ}{\mathsf{Occ}}

\newcommand{\Dev}{\mathsf{Dev}}
\newcommand{\last}{\mathsf{last}}
\newcommand{\first}{\mathsf{first}}

\newcommand{\Sat}{\textsc{Sat}}
\newcommand{\coSat}{\textsc{coSat}}

\renewcommand{\P}{\mathbb{P}}
\newcommand{\C}{\mathbb{C}}
\renewcommand{\S}{\mathbb{S}}
\newcommand{\A}{\mathbb{A}}
\renewcommand{\O}{\mathbb{O}}

\newcommand{\M}{\mathcal{M}}

\newcommand{\bx}{\bar{x}}
\newcommand{\by}{\bar{y}}

\newcommand{\bc}{\bar{c}}
\newcommand{\bL}{\overline{L}}

\newcommand{\dpi}{\dot{\pi}}
\newcommand{\dchi}{\dot{\chi}}

\newcommand{\trho}{\tilde{\rho}}
\newcommand{\teta}{\tilde{\eta}}
\renewcommand{\th}{\tilde{h}}

\newcommand{\hkappa}{\hat{\kappa}}

\newcommand{\NP}{\mathbf{NP}}

\newcommand{\coNP}{\mathbf{coNP}}
\newcommand{\PSpace}{\mathbf{PSpace}}
\newcommand{\ExpTime}{\mathbf{ExpTime}}
\newcommand{\BH}{\mathbf{BH}}
\newcommand{\X}{\mathbf{X}}
\newcommand{\U}{\mathbf{U}}
\newcommand{\G}{\mathbf{G}}
\newcommand{\F}{\mathbf{F}}
\newcommand{\B}{\mathbf{B}}

\newtheorem{thm}{Theorem}

\newtheorem{lm}{Lemma}

\theoremstyle{definition}
\newtheorem{defi}{Definition}
\newtheorem{pb}{Problem}

\theoremstyle{remark}
\newtheorem*{rk}{Remark}
\newtheorem{ex}{Example}

\author{Léonard Brice}
{Université Gustave Eiffel, France, and Université Libre de Bruxelles, Belgium}{lnrd.brice@gmail.com}{}{}

\author{Jean-François Raskin}
{Université Libre de Bruxelles, Belgium}{jraskin@ulb.ac.be}{}{}

\author{Marie van den Bogaard}
{LIGM, Univ. Gustave Eiffel, CNRS, F-77454 Marne-la-Vall\'{e}e, France}{marie.van-den-bogaard@univ-eiffel.fr}{}{}

\authorrunning{L. Brice, J.-F. Raskin, and M. van den Bogaard}

\Copyright{Léonard Brice, Jean-François Raskin, Marie van den Bogaard}

\ccsdesc{Software and its engineering: Formal methods; Theory of
	computation: Logic and verification; Theory of computation: Solution concepts in game theory.}

\keywords{Games on graphs, subgame-perfect equilibria, parity objectives.}

\funding{This work is partially supported by the ARC project Non-Zero Sum Game Graphs: Applications to Reactive Synthesis and Beyond (F\'{e}d\'{e}ration Wallonie-Bruxelles), the EOS project Verifying Learning Artificial Intelligence Systems (F.R.S.-FNRS \& FWO), the COST Action 16228 GAMENET (European Cooperation in Science and Technology), and by the PDR project \emph{Subgame perfection in graph games} (F.R.S.- FNRS).}
\acknowledgements{We wish to thank anonymous reviewers for their useful observations.}

\title{On the Complexity of SPEs in Parity Games}
\date{2021}

\begin{document}
	
	\maketitle
	
	\begin{abstract}
		We study the complexity of problems related to subgame-perfect equilibria (SPEs) in infinite duration non zero-sum multiplayer games played on finite graphs with parity objectives. 
		We present new complexity results that close gaps in the literature. Our techniques are based on a recent characterization of SPEs in prefix-independent games that is grounded on the notions of requirements and negotiation, and according to which the plays supported by SPEs are exactly the plays consistent with the requirement that is the least fixed point of the negotiation function.
		The new results are as follows. First, checking that a given requirement is a fixed point of the negotiation function is an $\NP$-complete problem.
		Second, we show that the SPE constrained existence problem is $\NP$-complete, this problem was previously known to be $\ExpTime$-easy and $\NP$-hard.
		Third, the SPE constrained existence problem is fixed-parameter tractable when the number of players and of colors are parameters.
		Fourth, deciding whether some requirement is the least fixed point of the negotiation function is complete for the second level of the Boolean hierarchy. Finally, the SPE-verification problem --- that is, the problem of deciding whether there exists a play supported by a SPE that satisfies some LTL formula --- is $\PSpace$-complete, this problem was known to be $\ExpTime$-easy and $\PSpace$-hard.
	\end{abstract}
	
	\section{Introduction}
	\emph{Nash equilibrium} (NE) is one of the central concepts from game theory to formalize the notion of rationality.
	It describes profiles of strategies in which no player has an incentive to change their strategy unilaterally.
	However, in sequential games, like games played on graphs, NEs are known to be plagued by {\em non-credible threats}: players can threaten other players in subgames with {\em non-rational actions} in order to force an equilibrium that avoids these subgames.
	To avoid non-credible threats, subgame-perfect equilibria are used instead. \emph{Subgame-perfect equilibria} (SPEs) are NEs that are NEs in all subgames of the original game: the players must act rationally in all subgames even after a deviation by another player.
	
	In this paper, we study the complexity of decision problems related to SPEs in sequential games played on graphs with parity objectives. In such a game, each vertex of the game graph has one color per player, and each player wants the least color they see infinitely often along a play, which is an infinite path in the graph, to be even.
	Parity conditions, in games as well as in automata, are canonical ways to represent $\omega$-regular constraints. It is known that SPEs always exist in parity games, as shown  in \cite{DBLP:conf/fsttcs/Ummels06}. Unfortunately, the precise complexity of the SPE constrained existence problem, i.e. the problem of deciding whether there exists an SPE that generates payoffs between two given thresholds, is left open in the literature: it is known to be $\ExpTime$-easy and $\NP$-hard. We prove here that it is in fact $\NP$-complete, and we provide several other new complexity results on related problems of interest.
	
	While previous attempts to solve this decision problem
	were based on alternating tree automata (\cite{GU08}), we obtain the new tight complexity results starting from concepts that we have introduced recently in \cite{Concur} to capture SPEs in mean-payoff games: the notions of \emph{requirements} and \emph{negotiation}.
	A requirement is a function $\lambda$ that maps each vertex $v$ of the game graph to a real value, that represents the lowest payoff that the player controlling $v$ should accept when facing other rational players.
	A play $\rho = \rho_0 \rho_1 \dots$ is \emph{$\lambda$-consistent} if for each vertex $\rho_k$, the player controlling $\rho_k$ gets at least the payoff $\lambda(\rho_k)$ in $\rho$.
	In Boolean games, such as parity games, we naturally consider requirements whose values are either $0$ or $1$ ($1$ meaning that the player must achieve their objective, $0$ that they may not).
	
	The negotiation function maps a requirement $\lambda$ to a requirement $\nego(\lambda)$, which captures from any vertex $v$ the maximal payoff that the corresponding player can ensure, against \emph{$\lambda$-rational} players, that is, players who play in such a way that they obtain at least the payoff specified by $\lambda$.
	Clearly, if $\lambda_0$ maps each vertex to $0$, then every play is $\lambda_0$-consistent.
	Then, the requirement $\nego(\lambda_0)$ maps each vertex $v$ to its \emph{antagonistic value}, i.e. the best payoff that the player controlling $v$ can ensure against an adversarial coalition of the other players (as any behavior of the other players is $\lambda_0$-rational).
	It is the case that the $\nego(\lambda_0)$-consistent plays are exactly the plays supported by NEs.
	
	But then, the following natural question is: given $v$ and $\lambda$, can the player who controls $v$ improve his worst-case value, if only plays that are consistent with $\lambda$ are proposed by the other players? Or equivalently, can this player enforce a better value when playing against players that are not willing to give away their own worst-case value? which is clearly a minimal goal for any rational adversary. So $\nego(\lambda)(v)$ returns this value; and this reasoning can be iterated. In~\cite{Concur}, it is shown that the least fixed point $\lambda^*$ of the negotiation function is exactly characterizing the set of plays supported by SPEs, for all prefix independent payoff functions with steady negotiation (which is the case for parity objectives) .
	
	
	Using that characterization of SPEs, we prove that SPEs always exist in parity games (Theorem~\ref{thm_existence_spe}).
	That result had already been proved by Ummels in \cite{DBLP:conf/fsttcs/Ummels06}: we use the concepts of requirements and negotiation to rephrase his proof in a more succinct way.
	
	\paragraph*{Main contributions}
		
	In order to get tight complexity results, we establish the links between the negotiation function and a class of zero-sum two-player games, the \emph{abstract negotiation games} (Theorem~\ref{thm_abstract_game}).
	Those games are played on an infinite arena, but we show that the players can play simple strategies that have polynomial size representation, while still playing optimally (Lemma~\ref{lm_reduced_strategy}).
	We show that a non-deterministic polynomial algorithm can decide which player has a winning strategy in that game, i.e. can decide whether $\nego(\lambda)(v) = 0$ or $1$, for a given requirement $\lambda$ on a given vertex in a given game: as a consequence, deciding whether the requirement $\lambda$ is a fixed point of the negotiation function is $\NP$-easy (Lemma~\ref{lm_fixed_point_np_easy}).
	We also show that the computation of $\lambda^*$, and consequently the SPE constrained existence problem, are fixed-parameter tractable if we fix the number of players and colors in the game (Theorem~\ref{thm_fpt}).
	
	Those algorithms can be exploited to obtain upper bounds for several problems related to SPEs.
	Most classical of them, the problem of deciding whether there exists an SPE generating a payoff vector between two given thresholds (SPE constrained existence problem), is $\NP$-complete (Theorem~\ref{thm_constrained_existence_np_complete}).
	That problem can be solved without computing the least fixed point of the negotiation function: such a problem is $\BH_2$-complete, as well as its decisional version --- i.e. deciding whether a given requirement $\lambda$ is equal to $\lambda^*$ (Theorem~\ref{thm_lfp_bh2_complete}).
	Finally, deciding whether there exists an SPE generating a play that satisfies some LTL formula (SPE-verification problem) is $\PSpace$-complete (Theorem~\ref{thm_spe_verif_pspace_complete}).

	\paragraph*{Related works}
	In \cite{GU08}, Ummels and Grädel solve the SPE constrained existence problem in parity games, and prove that such games always contain SPEs.
	Their algorithm is based on the construction of alternating tree automata, on which one can solve the emptiness problem in exponential time.
	The SPE constrained existence problem is therefore $\ExpTime$-easy, which was, to our knowledge, the best upper bound existing in the literature.
	The authors also prove the $\NP$-hardness of that problem.
	In what follows, we prove that it is actually $\NP$-complete.
	
	In \cite{DBLP:journals/mor/FleschP17}, Flesch and Predtetchinski present a general non-effective procedure to characterize the plays supported by an SPE in games with finitely many possible payoff vectors, as parity games.
	That characterization uses
	the abstract negotiation game, but does not use the notions of requirements and negotiation, and as a consequence does not yield an effective algorithm --- their procedure requires to solve infinitely many games that have an uncountable state space.
	
	In \cite{Concur}, we solve the SPE constrained existence problem on mean-payoff games.
	To that end, we define the notions of requirements and negotiation, and highlight the links between negotiation and the abstract negotiation game.
	Part of our results can be applied to every prefix-independent game with steady negotiation, which includes parity games.
	But, in order to get tight complexity results, we need here to introduce new notions, such as reduced strategies and deviation graphs.
	
	In \cite{DBLP:conf/concur/BrihayeBGRB19}, Brihaye et al. prove that the SPE constrained existence problem on quantitative reachability games is $\PSpace$-complete.
	Their algorithm updates continuously a function that heralds the notion of requirement, until it reaches a fixed point, that we can interpret as the least fixed point of the negotiation function.

	In \cite{DBLP:conf/lfcs/BrihayePS13}, Brihaye et al. give a characterization of NEs in cost-prefix linear games, based on the worst-case value.
	Parity games are cost-prefix linear, and the worst-case value is 
	captured by our notion of requirement.
	The authors do not study the notion of SPE in their paper.
	
	In \cite{thesis_noemie}, Meunier proposes a method to decide the existence of SPEs generating a given payoff, proving that it is equivalent to decide which player has a winning strategy in a Prover-Challenger game.
	That method could be used with parity games, but it would not lead to a better complexity than \cite{DBLP:conf/fsttcs/Ummels06}, and so it would not yield our $\NP$-completeness result.
	
	The applications of non-zero sum infinite duration games targeting reactive synthesis problems have gathered significant attention during the recent years, hence a rich literature on that topic.
	The interested reader may refer to the surveys \cite{DBLP:conf/lata/BrenguierCHPRRS16,DBLP:conf/dlt/Bruyere17} and their references.

	\paragraph*{Structure of the paper}
	
	In Section~\ref{sec_background}, we introduce the necessary background.
	In Section~\ref{sec_nego}, we present the concepts of requirements, negotiation and abstract negotiation game, and show how they can be applied to characterize SPEs in parity games.
	In Section~\ref{sec_algo}, we turn that characterization into algorithms that solve the aforementioned problems, and deduce upper bounds for their complexities.
	In Section~\ref{sec_lower_bounds}, we match them with lower bounds, and conclude on the precise complexities of those problems.

	\section{Background} \label{sec_background}
	
	In the sequel, we use the word \emph{game} for Boolean turn-based games played on finite graphs.
	
	\begin{defi}[Game]
		A \emph{game} is a tuple $G = (\Pi, V, (V_i)_{i \in \Pi}, E, \mu)$, where:
		\begin{itemize}
			\item $\Pi$ is a finite set of \emph{players};
			
			\item $(V, E)$ is a finite directed graph, whose vertices and edges are also called \emph{states} and \emph{transitions}, and in which every state has at least one outgoing transition;
			
			\item $(V_i)_{i \in \Pi}$ is a partition of $V$, where each $V_i$ is the set of the states \emph{controlled} by player $i$;
			
			\item $\mu: V^\omega \to \{0, 1\}^\Pi$ is a \emph{payoff function}, which maps each sequence of states $\rho$ to the tuple $\mu(\rho) = (\mu_i(\rho))_i$ of the players' payoffs: player $i$ \emph{wins} $\rho$ if $\mu_i(\rho) = 1$, and \emph{loses} otherwise.
		\end{itemize}
	\end{defi}
	
	A play in such a game can be seen as an infinite sequence of moves of a token on the graph $(V, E)$: when the token is on a given vertex, the player controlling that vertex chooses the edge along which it will move.
	And then, the player controlling the next vertex chooses where it moves next, and so on.
	
	\begin{defi}[Play, history]
		A \emph{play} (resp. \emph{history}) in the game $G$ is an infinite (resp. finite) path in the graph $(V, E)$.
		We write $\Plays G$ (resp. $\Hist G$) for the set of plays (resp. histories) in $G$.
		We write $\Hist_i G$ for the set of histories in $G$ of the form $hv$, where $v \in V_i$.
		We write $\Occ(\rho)$ (resp. $\Occ(h)$) for the set of vertices that occur at least once in the play $\rho$ (resp. the history $h$), and $\Inf(\rho)$ for the set of vertices that occur infinitely often in $\rho$.
		We write $\first(h)$ (resp. $\first(\rho)$) the first vertex of $h$ (resp. $\rho$), and $\last(h)$ its last vertex.
	\end{defi}
	
	Often, we need to specify an initial state for a game.
	
	\begin{defi}[Initialized game]
		An \emph{initialized game} is a pair $(G, v_0)$, often written $G_{\|v_0}$, where $G$ is a game and $v_0 \in V$ is the \emph{initial vertex}.
		A play (resp. history) of $G$ is a play (resp. history) of $G_{\|v_0}$ iff its first state is $v_0$.
		We write $\Plays G_{\|v_0}$ (resp. $\Hist G_{\|v_0}$, $\Hist_i G_{\|v_0}$) for the set of plays (resp. histories, histories ending in $V_i$) in $G_{\|v_0}$.
	\end{defi}
	
	When the context is clear, we call \emph{game} both non-initialized and initialized games.
	
	\begin{defi}[Strategy, strategy profile]
		A \emph{strategy} for player $i$ in $G_{\|v_0}$ is a function $\sigma_i: \Hist_i G_{\|v_0} \to V$ such that for each history $hv \in \Hist_i G_{\|v_0}$, we have $v \sigma_i(hv) \in E$.
		
		A \emph{strategy profile} for $P \subseteq \Pi$ is a tuple $\bsigma_P = (\sigma_i)_{i \in P}$ where each $\sigma_i$ is a strategy for player $i$.
		When $P = \Pi$, the strategy profile is \emph{complete}, and we usually write it $\bsigma$.
		For each $i \in \Pi$, we write $-i$ for the set $\Pi \setminus \{i\}$.
		When $\btau_P, \btau'_Q$ are two strategy profiles with $P \cap Q = \emptyset$, we write $(\btau_P, \btau'_Q)$ the strategy profile $\bsigma_{P \cup Q}$ defined by $\sigma_i = \tau_i$ if $i \in P$, and $\sigma_i = \tau'_i$ if $i \in Q$.
		We write $\Sigma_i G_{\|v_0}$ (resp. $\Sigma_P G_{\|v_0}$) the set of all strategies (resp. strategy profiles) for player $i$ (resp. the set $P$) in $G_{\|v_0}$.
		
		A history or a play is \emph{compatible with} (or \emph{supported by}) a strategy $\sigma_i$ if for each of its prefixes $hv$ with $h \in \Hist_i G$, we have $v = \sigma_i(h)$.
		It is compatible with a strategy profile $\bsigma_P$ if it is compatible with $\sigma_i$ for each $i \in P$.
		When a strategy profile $\bsigma$ is complete, there is one unique play in $G_{\|v_0}$ that is compatible with it, written $\< \bsigma \>_{v_0}$ and called the \emph{outcome} of $\bsigma$.
		
		A strategy $\sigma_i$ is \emph{memoryless} when for each state $v$ and every two histories $h$ and $h'$, we have $\sigma_i(hv) = \sigma_i(h'v)$.
		In that case, we liberally consider that $\sigma_i$ is defined from every state, and write $\sigma_i(v)$ for every $\sigma_i(hv)$.
	\end{defi}
	
	Before defining the notion of SPEs, we need to define a weaker, but more classical, solution concept: Nash equilibria.
	A Nash equilibrium is a strategy profile such that no player can improve their payoff by deviating unilaterally from their strategy.
	
	\begin{defi}[Nash equilibrium]
		A complete strategy profile $\bsigma$ in $G_{\|v_0}$ is a \emph{Nash equilibrium} --- or \emph{NE} for short --- iff for each player $i$ and for every strategy $\sigma'_i$, we have $\mu_i(\< \bsigma_{-i}, \sigma'_i \>_{v_0})~\leq~\mu_i(\< \bsigma \>_{v_0})$.
	\end{defi}
	
	An SPE is an NE in all the subgames, in the following formal sense.
	
	\begin{defi}[Subgame, substrategy]
		Let $hv$ be a history in $G_{\|v_0}$.
		The \emph{subgame} of $G$ after $hv$ is the initialized game $G_{\|hv} = (\Pi, V, (V_i)_i, E, \mu_{\|hv})_{\|v}$, where $\mu_{\|hv}$ maps each play to its payoff in $G$, assuming that the history $hv$ has already been played: formally, for every $\rho \in \Plays G_{\|hv}$, we have $\mu_{\|hv}(\rho) = \mu(h\rho)$.
		If $\sigma_i$ is a strategy in $G_{\|v_0}$, its \emph{substrategy} after $hv$ is the strategy $\sigma_{i\|hv}$ in $G_{\|hv}$, defined by $\sigma_{i\|hv}(h') = \sigma_i(hh')$ for every $h' \in \Hist_i G_{\|hv}$.
	\end{defi}
	
	\begin{defi}[Subgame-perfect equilibrium]
		A complete strategy profile $\bsigma$ in $G_{\|v_0}$ is a \emph{subgame-perfect equilibrium} --- or \emph{SPE} for short --- iff for every history $hv$ in $G_{\|v_0}$, the substrategy profile $\bsigma_{\|hv}$ is a Nash equilibrium.
	\end{defi}
	
	Throughout this paper, we mostly study \emph{parity} games.
	
	\begin{defi}[Parity game]
		The game $G$ is a \emph{parity game} if there exists a tuple of \emph{color functions} $(\kappa_i:~V~\to~\N)_{i \in \Pi}$, such that each play $\rho$ is won by a given player $i$ --- i.e. $\mu_i(\rho) = 1$ --- iff the least color seen infinitely often by player $i$, i.e. the integer $\min \kappa_i(\Inf(\rho))$, is even.
		
		A \emph{Büchi} game is a parity game where all colors are either $0$ or $1$ --- or equivalently, a game in which the objective of each player is to visit infinitely often a given set of vertices.
		A \emph{coBüchi} game is a parity game where all colors are either $1$ or $2$ --- or equivalently, a game in which the objective of each player is to eventually avoid a given set of vertices.
	\end{defi}

	\begin{ex}
		Consider the (coBüchi) game represented by Figure~\ref{fig_ex1}: both players win the play $ace^\omega$, and lose any other.
		A first NE in that game is the strategy profile in which both players always go to the right: its outcome is $ace^\omega$, which is won by both players, hence none can strictly improve their payoff by deviating.
		A second NE is the strategy profile in which both players always go down: its outcome is $ab^\omega$, which is lost by both players.
		However, player $\Box$ cannot improve his strategy, because he never plays; and player $\Circle$ cannot neither, because if she goes right, then $\Box$ plans to go down, and she still loses.
		Only the first one is an SPE: for player $\Box$, planning to go down from the state $c$ is a non-credible threat.
		
		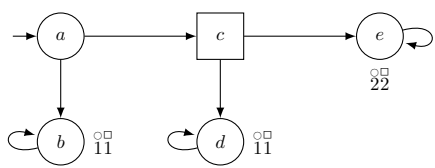
\begin{figure}
			\begin{center}
				\begin{tikzpicture}[->,>=latex,scale=0.7, every node/.style={scale=0.7},initial text={}]
				\node[state, initial left] (a) at (0, 0) {$a$};
				\node[state] (b) at (0, -2) {$b$};
				\node[state, rectangle] (c) at (3, 0) {$c$};
				\node[state] (d) at (3, -2) {$d$};
				\node[state] (e) at (6, 0) {$e$};
				
				\path (a) edge (b);
				\path (b) edge[loop left] (b);
				\path (a) edge (c);
				\path (c) edge (d);
				\path (d) edge[loop left] (d);
				\path (c) edge (e);
				\path (e) edge[loop right] (e);
				
				\node (b') at (0.8, -2) {$\stackrel{\playcircle}{1} \stackrel{\Box}{1}$};
				\node (d') at (3.8, -2) {$\stackrel{\playcircle}{1} \stackrel{\Box}{1}$};
				\node (e') at (6, -0.8) {$\stackrel{\playcircle}{2} \stackrel{\Box}{2}$};
				\end{tikzpicture}
			\end{center}
			\caption{A coBüchi game with two NEs and one SPE.}
			\label{fig_ex1}
		\end{figure}
	\end{ex}
	
	An important property of parity games is that they are \emph{prefix-independent}.
	
	\begin{defi}[Prefix-independent game]
		The game $G$ is \emph{prefix-independent} iff for every history $h$, we have $\mu_{\|h} = \mu$ --- or, equivalently, $G_{\|h} = G_{\|\last(h)}$.
	\end{defi}

	In such games, we search algorithms that solve the following problems.
	Let us specify that in all the sequel, tuples, as well as mappings, are ordered by the componentwise order.
	
	\begin{pb}[SPE constrained existence problem]
		Given a parity game $G_{\|v_0}$ and two thresholds $\bx, \by \in \{0, 1\}^\Pi$, is there an SPE $\bsigma$ in $G_{\|v_0}$ such that $\bx \leq \mu(\< \bsigma \>_{v_0}) \leq \by$?
	\end{pb}
	
	
	The next problem requires a definition of the linear temporal logic, LTL.
	
	\begin{defi}[LTL formulas]
		The \emph{linear temporal logic} --- or \emph{LTL} for short --- over the set of atomic propositions $\A$ is defined as follows: syntactically, each $a \in \A$ is an LTL formula, and if $\phi$ and $\psi$ are LTL formulas, then $\neg \phi$, $\phi \vee \psi$, $\X \phi$, and $\phi \U \psi$ are LTL formulas.
			
		Semantically, if $\nu = \nu_0 \nu_1 \dots$ is an infinite sequence of valuations of $\A$, then:
			\begin{itemize}
				\item $\nu \models a$ iff $\nu_0(a) = 1$;
				
				\item $\nu \models \neg \phi$ iff $\nu \not\models \phi$;
				
				\item $\nu \models \phi \vee \psi$ iff $\nu \models \phi$ or $\nu \models \psi$;
				
				\item $\nu \models \X \phi$ iff $\nu_1 \nu_2 \dots \models \phi$;
				
				\item $\nu \models \phi \U \psi$ iff there exists $k \in \N$ such that $\nu_k \nu_{k+1} \dots \models \psi$, and for each $\l < k$, we have $\nu_\l \nu_{\l+1} \dots \models \phi$.
			\end{itemize}
	\end{defi}
	
	We will also make use of the classical notations $\wedge$, $\top$, $\bot$, $\Rightarrow$, $\F$ or $\G$ defined as abbreviations using the symbols chosen here as primitives.
	In particular, we write $\top$ for $a \vee \neg a$, $\F \phi$ ("finally $\phi$") for $\top \U \phi$, and $\G \phi$ ("globally $\phi$") for $\neg \F \neg \phi$.
	When we use LTL to describe plays in a game, w.l.o.g. and for simplicity, the atom set is $\A = V$, and each play $\rho$ is assimilated to the sequence of valuations $\nu$ defined by $\nu_k(v) = 1$ iff $\rho_k = v$.
	For example, when $u$ and $v$ are two vertices, the play $(uv)^\omega$ is the only play satisfying the formula $u \wedge \mathbf{G} \left( (u \Rightarrow \X v) \wedge (v \Rightarrow \X u) \right)$.
	
	\begin{pb}[SPE-verification problem]
		Given a parity game $G_{\|v_0}$ and an LTL formula $\phi$, is there an SPE $\bsigma$ in $G_{\|v_0}$ such that $\< \bsigma \>_{v_0} \models \phi$?
	\end{pb}
	
	
\begin{rk}
    The natural problems of deciding whether \emph{all} SPEs generate a payoff vector between two thresholds, or a play that satisfies some LTL formula, are the duals of the aforementioned problems, and their complexities are obtained as direct corollaries.
    For example, in a given parity game, all the outcomes of SPEs satisfy the formula $\phi$ if and only if there does not exist an SPE whose outcome satisfies $\neg \phi$.
    Since we will show that the SPE-verification problem is $\PSpace$-complete, its dual will also be $\PSpace$-complete. Similarly, the SPE constrained universality problem is $\coNP$-complete as we will show that its dual, the SPE constrained existence problem, is $\NP$-complete.
\end{rk}
	
While we do not recall the definition of classical complexity classes here (such as $\NP$, $\coNP$, or $\PSpace$, see \cite{DBLP:books/daglib/0018514}), we recall the definition of the class $\BH_2$: the second level of the Boolean hierarchy.
For more details about the Boolean hierarchy itself, see \cite{DBLP:conf/fct/Wechsung85}.

	\begin{defi}[Class $\BH_2$]
		The complexity class $\BH_2$ is the class of problems of the form $P \cap Q$, where $P$ is $\NP$-easy and $Q$ is $\coNP$-easy, and both have the same set of instances.
		In other words, a $\BH_2$-easy problem is a problem that can be decided with one call to an $\NP$ algorithm, and one to a $\coNP$ algorithm.
	\end{defi}
\noindent

\begin{rk}
    The class $\BH_2$ must not be mistaken with the class $\NP \cap \coNP$, that gathers the problems that can be solved by an $\NP$ algorithm \emph{as well as} by a $\coNP$ one: the latter is included in $\NP$ and in $\coNP$, while the former contains them.
\end{rk}

To attach intuition to this definition, let us present a useful $\BH_2$-complete problem.

	\begin{pb}[$\Sat \times \coSat$]
		Given a pair $(\phi_1, \phi_2)$ of propositional logic formulas, is it true that $\phi_1$ is satisfiable and that $\phi_2$ is not?
	\end{pb}
	
	\begin{lm}[App.~\ref{pf_sat_cosat}] \label{lm_sat_cosat}
		The problem $\Sat \times \coSat$ is $\BH_2$-complete.
	\end{lm}

	\section{Negotiation in parity games} \label{sec_nego}
	
	\subsection{Requirements, negotiation, and link with SPEs}
	
	In the algorithms we provide, we make use of the characterization of SPEs in prefix-independent games that has been presented in \cite{Concur}.
	We recall here the notions needed, slightly adapted to Boolean games.
	
	\begin{defi}[Requirement]
		A \emph{requirement} on a game $G$ is a mapping $\lambda:~V~\to~\{0, 1, +\infty\}$.
		The set of the requirements on $G$ is denoted by $\Req G$, and is ordered by the componentwise order $\leq$: we write $\lambda \leq \lambda'$ when we have $\lambda(v) \leq \lambda'(v)$ for all $v$.
	\end{defi}
	
	\begin{defi}[$\lambda$-consistency]
		Let $\lambda$ be a requirement on the game $G$.
		A play $\rho$ in $G$ is \emph{$\lambda$-consistent} iff for each player $i$ and every index $k$ such that $\rho_k \in V_i$, we have $\mu_i(\rho_k\rho_{k+1} \dots) \geq \lambda(\rho_k)$.
		The set of $\lambda$-consistent plays in $G_{\|v_0}$ is denoted by $\lCons(v_0)$.
	\end{defi}

	In this paper, we will consider specifically requirements that are \emph{satisfiable}.
	
	\begin{defi}[Satisfiability]
	    The requirement $\lambda$, on the game $G$, is \emph{satisfiable} iff for each state $v \in V$, there exists at least one $\lambda$-consistent play from $v$.
	\end{defi}

\begin{lm}\label{lm_satisfiable_np_easy}
    Given a parity game $G$ and a requirement $\lambda$, deciding whether $\lambda$ is satisfiable is $\NP$-easy.\footnote{It is actually $\NP$-complete, as we can prove by slightly adapting the proof of Theorem~\ref{thm_fixed_point_np_complete}.}
\end{lm}

\begin{proof}
    An $\NP$ algorithm for that problem guesses, first, a family $(h_v, W_v)_{v \in V}$, where for each $v$, $h_v$ is a history without cycle starting from $v$ and $W_v$ is a subset of $V$ with $\last(h_v) \in W_v$.
    That family is an object of polynomial size, and certifies that $\lambda$ is satisfiable if for each $v$: (1) $\lambda(v) \neq +\infty$; (2) the subgraph $(W_v, E \cap W_v^2)$ is strongly connected; (3) for each vertex $u \in \Occ(h_v) \cup W_v$ such that $\lambda(u) = 1$, if $i$ is the player who controls $u$, then the color $\min~\kappa_i(W_v)$ is even.
    
    Indeed, if those three points are satisfied, then for each $v$, the play $h_v c_v^\omega$, where $c_v$ is a cycle (not necessarily simple, but which can be chosen of size at most $(\card W_v)^2$) that visits all the vertices of $W_v$ at least once and none other, is $\lambda$-consistent.
    Conversely, if from each $v$, there exists a $\lambda$-consistent play $\rho$, then the pair $(h_v, W_v)$, where $W_v = \Inf(\rho)$ and $h_v$ is a prefix of $\rho$ ending in $W_v$ in which the cycles have been removed, satisfies those three properties --- those can be checked in polynomial time.
\end{proof}

Each requirement induces a notion of rationality for a coalition of players.
	
	\begin{defi}[$\lambda$-rationality]
		Let $\lambda$ be a requirement on the game $G$.
		A strategy profile $\bsigma_{-i}$ is \emph{$\lambda$-rational assuming} the strategy $\sigma_i$ iff for every history $hv$ compatible with $\bsigma_{-i}$, the play $\< \bsigma_{\|hv} \>_v$ is $\lambda$-consistent.
		It is \emph{$\lambda$-rational} if it is $\lambda$-rational assuming some strategy.
		The set of $\lambda$-rational strategy profiles in $G_{\|v_0}$ is denoted by $\lRat(v_0)$.
	\end{defi}
	
	The notion of $\lambda$-rationality qualifies the \emph{environment} against player $i$, i.e. the coalition of all the players except $i$: they play $\lambda$-rationally if their strategy profile can be completed by a strategy of player $i$, such that in every subgame, each player gets their requirement satisfied.
	
	But then, $\lambda$-rationality restrains the behaviours of the players against player $i$: that one may be able to win against a $\lambda$-rational environment while it is not the case against a fully hostile one.
	This is what the \emph{negotiation function} captures.
	
	\begin{defi}[Negotiation]
		The \emph{negotiation function} is a function that transforms every requirement $\lambda$ into a requirement $\nego(\lambda)$, defined by, for each $i \in \Pi$ and $v \in V_i$, and with the convention $\inf\emptyset = +\infty$:
		$$\nego(\lambda)(v) = \inf_{\bsigma_{-i} \in \lRat(v)} ~\sup_{\sigma_i \in \Sigma_i(G_{\|v})} ~\mu_i(\< \bsigma_{-i}, \sigma_i \>_v).$$
	\end{defi}
	
	\begin{rk}
	    If player $i$ follows the strategy $\sigma_i$ assuming which $\bsigma_{-i}$ is $\lambda$-rational, then they get at least the payoff $\lambda(v)$, hence $\nego(\lambda)(v) \geq \lambda(v)$ and the negotiation function is non-decreasing.
	    Moreover, if $\lambda \leq \lambda'$, then all the $\lambda'$-rational strategy profiles are also $\lambda$-rational, hence $\nego(\lambda) \leq \nego(\lambda')$ and the negotiation function is monotonic.
	\end{rk}

	The fixed points of the negotiation function characterize the SPEs of a game: indeed, when some play $\rho$ is $\lambda$-consistent for some fixed point $\lambda$, it means that it is won by every player who could ensure their victory from a state visited by $\rho$, while playing against a rational environment.
	Better: all the SPEs are characterized by the least fixed point of the negotiation function, which exists by Tarski's fixed point theorem, and which we will write $\lambda^*$ in the rest of this paper.
	An equivalent result exists for NEs, that are characterized by the requirement $\nego(\lambda_0)$, where $\lambda_0: v \mapsto 0$ is the \emph{vacuous requirement}.
	
	\begin{thm} \label{thm_nego}
		In a prefix-independent Boolean game $G_{\|v_0}$:
		\begin{itemize}
		    \item the set of NE outcomes is exactly the set of $\nego(\lambda_0)$-consistent plays;
		    
		    \item the set of SPE outcomes is exactly the set of $\lambda^*$-consistent plays.
		\end{itemize}
	\end{thm}
	
	\begin{proof}
		By \cite{Concur}, this result is true for any prefix-independent game \emph{with steady negotiation}, i.e. such that for every requirement $\lambda$, for every player $i$ and for every vertex $v$, if there exists a $\lambda$-rational strategy profile $\bsigma_{-i}$ from $v$, there exists one that minimizes the quantity $\sup_{\sigma_i}~\mu_i(\< \bsigma_{-i}, \sigma_i \>_v).$
		In the case of Boolean games, the function $\mu_i$ can only take the values $0$ and $1$, hence this supremum is always realized.
	\end{proof}

\begin{ex}
	Let us consider again the game of Figure~\ref{fig_ex1}.
	Every play in that game --- like in every game --- is $\lambda_0$-rational.
	The requirement $\lambda_1 = \nego(\lambda_0)$ is equal to $1$ on the states $c$ and $e$ (the states from which the player controlling those states can enforce the victory), and to $0$ in each other one.
	Then, the $\lambda_1$-consistent plays are exactly the plays supported by a Nash equilibrium: the play $ace^\omega$, and the play $ab^\omega$.
	
	Now, from the state $a$, the only strategy profile that can make player $\Circle$ lose if she chooses to go to $c$ is $\sigma_\Box: ac \mapsto d$, which was $\lambda_0$-rational but is not $\lambda_1$-rational: the play $cd^\omega$ is not $\lambda_1$-consistent.
	Therefore, against a $\lambda_1$-rational environment, player $\Circle$ can enforce the victory by going to the state $c$, hence $\lambda_2(a) = 1$, where $\lambda_2 = \nego(\lambda_1)$.
	Then, the requirement $\lambda_2$ is a fixed point of the negotiation function, and consequently the least one, hence the only play supported by an SPE from the state $a$ is the only play that is $\lambda_2$-consistent, namely $ace^\omega$.
\end{ex}

	\subsection{The existence of SPEs in parity games}
	
	In parity games, the existence of SPEs is guaranteed.

	\begin{thm}[\cite{DBLP:conf/fsttcs/Ummels06}] \label{thm_existence_spe}
		There exists an SPE in every parity game.
	\end{thm}
	
\begin{proof}[Proof sketch]
    This theorem is a result due to Ummels.
    In App.~\ref{pf_existence_spe}, we rephrase his proof in terms of requirements and negotiation.
    Let us give here the main intuitions.
    We define a decreasing sequence $(E_n)_n$ of subsets of $E$, and an associated sequence $(\lambda'_n)_n$ of requirements, keeping the hypothesis that $E_n$ always contains at least one outgoing edge from each vertex.
	First, $E_0 = E$ and $\lambda'_0$ is the vacuous requirement.
	Then, for every $n$, for each player $i$ and each $v \in V_i$, we define $\lambda'_{n+1}(v)$ as equal to $1$ if and only if in the game obtained from $G$ by removing the edges that are not in $E_n$, player $i$ can enforce the victory from $v$, against a fully hostile environment.
	Then, from each such state, we choose a memoryless winning strategy (which always exists, see \cite{DBLP:conf/dagstuhl/Kusters01}), that is, we choose one edge to always follow to ensure the victory, and we remove the other outgoing edges from $E_n$ to obtain $E_{n+1}$.
	We prove that for each $n$, we have $\lambda'_{n+1} \geq \nego(\lambda'_n)$, hence the sequence $(\lambda'_n)_n$ converges to a satisfiable fixed point of the negotiation function --- which is not necessarily the least one.
\end{proof}

	As a consequence, in every parity game $G$, we have $\lambda^*(v) \in \{0, 1\}$: this is why in what follows, we only consider requirements with values in $\{0, 1\}$.

\begin{ex}
	Let us consider the (Büchi) game of Figure~\ref{fig_ex2}: recall that the objective of each player is to see infinitely often the color $0$.
	If we follow the algorithm from \cite{DBLP:conf/fsttcs/Ummels06}, as presented above, we remove the edge $df$ --- because always going to $e$ is a winning strategy for player $\Diamond$ from $d$ --- and the edges $ba$ and $bd$ --- because always going to $c$ is a winning strategy for player $\Box$ from $b$.
	Then, the algorithm reaches a fixed point, see Figure~\ref{fig_ex2'}.
	Our proof states that every play that uses only the remaining edges is a play supported by an SPE.
	Indeed, those plays are $\lambda'_1$-consistent, where $\lambda'_1$ is the requirement given in red on Figure~\ref{fig_ex2'}: it is a fixed point of the negotiation function.
	However, it is not the least one: for example, the play $ab(dedf)^\omega$ is not $\lambda'_1$-consistent, but it is also a play that is supported by an SPE.
	The least fixed point is given in red on Figure~\ref{fig_ex2}.
	
	\begin{figure}
	    \centering
		\begin{subfigure}[b]{0.45\textwidth}
		    \centering
			\begin{tikzpicture}[->, >=latex,scale=0.7, every node/.style={scale=0.7}]
			\node[state] (a) at (-4, 0) {$a$};
			\node[state, rectangle] (b) at (-2, 0) {$b$};
			\node[state] (c) at (-2, -2) {$c$};
			\node[state, diamond] (d) at (0, 0) {$d$};
			\node[state] (e) at (1.4, -1.4) {$e$};
			\node[state] (f) at (1.4, 1.4) {$f$};
			
			\path[bend left = 20] (a) edge (b);
			\path[bend left = 20] (b) edge (a);
			\path[bend left = 20] (c) edge (b);
			\path[bend left = 20] (b) edge (c);
			\path[loop below] (c) edge (c);
			\path (b) edge (d);
			\path[bend left = 20] (d) edge (e);
			\path[bend left = 20] (e) edge (d);
			\path[bend left = 20] (d) edge (f);
			\path[bend left = 20] (f) edge (d);
			\path[loop below] (e) edge (e);
			\path[loop above] (f) edge (f);
			
			\node (a') at (-4, 0.8) {$\stackrel{\playcircle}{0} \stackrel{\Box}{1} \stackrel{\Diamond}{1}$};
			\node (b') at (-2, 0.8) {$\stackrel{\playcircle}{1} \stackrel{\Box}{1} \stackrel{\Diamond}{1}$};
			\node (c') at (-1.2, -2) {$\stackrel{\playcircle}{0} \stackrel{\Box}{0} \stackrel{\Diamond}{1}$};
			\node (d') at (-0.5, 0.5) {$\stackrel{\playcircle}{1} \stackrel{\Box}{1} \stackrel{\Diamond}{1}$};
			\node (e') at (2.2, -1.4) {$\stackrel{\playcircle}{1} \stackrel{\Box}{0} \stackrel{\Diamond}{1}$};
			\node (f') at (2.2, 1.4) {$\stackrel{\playcircle}{1} \stackrel{\Box}{1} \stackrel{\Diamond}{0}$};
			
			\node[red] (l) at (-4, -2) {$(\lambda^*)$};
			\node[red] (a'') at (-4, -0.7) {$0$};
			\node[red] (b'') at (-1.3, -0.3) {$1$};
			\node[red] (c'') at (-2.7, -2) {$1$};
			\node[red] (d'') at (-0.4, -0.4) {$1$};
			\node[red] (e'') at (0.7, -1.4) {$0$};
			\node[red] (f'') at (0.7, 1.4) {$0$};
			\end{tikzpicture}
			\caption{A Büchi game.}
			\label{fig_ex2}
		\end{subfigure}
		\begin{subfigure}[b]{0.45\textwidth}
		    \centering
			\begin{tikzpicture}[->, >=latex,scale=0.7, every node/.style={scale=0.7}]
			\node[state] (a) at (-4, 0) {$a$};
			\node[state, rectangle] (b) at (-2, 0) {$b$};
			\node[state] (c) at (-2, -2) {$c$};
			\node[state, diamond] (d) at (0, 0) {$d$};
			\node[state] (e) at (1.4, -1.4) {$e$};
			\node[state] (f) at (1.4, 1.4) {$f$};
			
			\path[bend left = 20] (a) edge (b);
			\path[bend left = 20] (b) edge (c);
			\path[loop below] (c) edge (c);
			\path[bend left = 20] (d) edge (f);
			\path[bend left = 20] (e) edge (d);
			\path[bend left = 20] (f) edge (d);
			\path[loop below] (e) edge (e);
			\path[loop above] (f) edge (f);
			
			\node (a') at (-4, 0.8) {$\stackrel{\playcircle}{0} \stackrel{\Box}{1} \stackrel{\Diamond}{1}$};
			\node (b') at (-2, 0.8) {$\stackrel{\playcircle}{1} \stackrel{\Box}{1} \stackrel{\Diamond}{1}$};
			\node (c') at (-1.2, -2) {$\stackrel{\playcircle}{0} \stackrel{\Box}{0} \stackrel{\Diamond}{1}$};
			\node (d') at (-0.5, 0.5) {$\stackrel{\playcircle}{1} \stackrel{\Box}{1} \stackrel{\Diamond}{1}$};
			\node (e') at (2.2, -1.4) {$\stackrel{\playcircle}{1} \stackrel{\Box}{0} \stackrel{\Diamond}{1}$};
			\node (f') at (2.2, 1.4) {$\stackrel{\playcircle}{1} \stackrel{\Box}{1} \stackrel{\Diamond}{0}$};
			
			\node[red] (l) at (-4, -2) {$(\lambda'_1)$};
			\node[red] (a'') at (-4, -0.7) {$1$};
			\node[red] (b'') at (-1.3, -0.3) {$1$};
			\node[red] (c'') at (-2.7, -2) {$1$};
			\node[red] (d'') at (-0.4, -0.4) {$1$};
			\node[red] (e'') at (0.7, -1.4) {$0$};
			\node[red] (f'') at (0.7, 1.4) {$0$};
			\end{tikzpicture}
			\caption{Fixed point as computed in \cite{DBLP:conf/fsttcs/Ummels06}.}
			\label{fig_ex2'}
		\end{subfigure}
		\label{fig_ex2''}
		\caption{}
	\end{figure}
\end{ex}
	
The interested reader will find in Appendix~\ref{app_ex} an additional example of parity game, on which we computed the iterations of the negotiation function.

	\subsection{Abstract negotiation game}
	
	Now, let us study how we can compute the negotiation function.
	The \emph{abstract negotiation game} is a tool which already appeared in~\cite{DBLP:journals/mor/FleschP17}, and which has been linked to the negotiation function in~\cite{Concur}.
	It is a game on an infinite graph that opposes two players, \emph{Prover} and \emph{Challenger}: Prover constructs a $\lambda$-rational strategy profile by proposing plays, and Challenger constructs player $i$'s response by accepting those plays or deviating from them.
	We slightly simplify the definition here, by considering only satisfiable requirements, which guarantees that Prover has always a play to propose\footnote{In~\cite{Concur}, a second sink state $\bot$ is added to enable Prover to give up when she has no play to propose. Another purely technical difference is the existence here of a mandatory transition from each state $[hv]$ to the state $[v]$, instead of letting Prover propose a play directly from the state $[hv]$: thus, there are few states from which Prover has a choice to make, which will be useful in what follows.}.

	\begin{defi}[Abstract negotiation game]
		Let $G$ be a parity game, let $\lambda$ be a satisfiable requirement, let $i \in \Pi$ and let $v_0 \in V_i$.
		The associated \emph{abstract negotiation game} is the two-player zero-sum game $\Abs_{\lambda i}(G)_{\|[v_0]} = \left(\{\P, \C\}, S, (S_\P, S_\C), \Delta, \nu \right)_{\|[v_0]}$ where:
		\begin{itemize}
			\item the players $\P$ and $\C$ are called respectively \emph{Prover} and \emph{Challenger};
			
			\item Challenger's states are of the form $[\rho]$, where $\rho$ is a $\lambda$-consistent play of $G$;
			
			\item Prover's states are of the form $[hv]$, where $h \in \Hist_i(G) \cup \{\epsilon\}$ and $\last(h)v \in E$, plus one additional sink state $\top$;
			
			\item the set $\Delta$ contains the transitions of the forms:
			\begin{itemize}
				\item $[v][\rho]$, where $\first(\rho) = v$: Prover proposes the play $\rho$;
				
				\item $[\rho][\rho_0 \dots \rho_k v]$, where $k \in \N$, $v \neq \rho_{k+1}$ and $\rho_kv \in E$: Challenger refuses and deviates;
				
				\item $[hv][v]$ with $h \neq \epsilon$: then, Prover has to propose a new play from the vertex $v$;
				
				\item $[\rho] \top$: Challenger accepts the proposed play;
				
				\item $\top\top$: the game is over;
			\end{itemize}
			
			\item When $\pi$ is a play in the abstract negotiation game, we will use the notation $\dpi$ to denote the play in the original game constructed by Prover's proposals and Challenger's deviations. Thus, the play $\pi$ is won by Challenger iff one of the following conditions is satisfied:
			
			\begin{itemize}
				\item the play $\pi$ has the form $[v_0][\rho^0][h^0 v_1][v_1][\rho^1] \dots [h^{n-1}v_n] [v_n] [\rho^n] \top^\omega,$
				i.e. Challenger accepts a play proposed by Prover, and the play $\dpi = h^0 \dots h^{n-1} \rho^n$ is won by player $i$;
				
				\item or the play $\pi$ has the form $[v_0][\rho^0][h^0 v_1][v_1][\rho^1][h^1v_2] \dots$,
				i.e. Challenger always deviates from the play proposed by Prover, and the play $\dpi = h^0 h^1 \dots$ is won by player $i$.
			\end{itemize}
		\end{itemize}
	\end{defi}

	\begin{thm}[\cite{Concur}, Appendix E] \label{thm_abstract_game}
		Let $G$ be a prefix-independant Boolean game, let $\lambda$ be a satisfiable requirement, let $i \in \Pi$ and let $v_0 \in V_i$. Then, we have $\nego(\lambda)(v_0) = 0$ if and only if Prover has a winning strategy in the associated abstract negotiation game.\footnote{The non-existence of the sink state $\bot$, in which Prover is supposed to get the payoff $-\infty$, does not change this result: since $\lambda$ is assumed to be satisfiable, Prover has always a strategy to get at least the payoff $0$, hence no optimal strategy of Prover plans to follow a transition to $\bot$.}
	\end{thm}

	\begin{ex}
		Let $G$ be the game from Figure~\ref{fig_ex1}.
		In this particular case, since there are finitely many possible plays, the abstract negotiation games $\Abs_{\lambda_0 \playcircle} G$ and $\Abs_{\lambda_1 \playcircle}$ have a finite state space.
		They are represented in Figure~\ref{fig_abstract}: the blue states are Prover's states, and the orange ones are Challenger's.
		The dashed states belong to $\Abs_{\lambda_0 \playcircle} G$ but not to $\Abs_{\lambda_1 \playcircle}$.
		Observe that Prover has a winning strategy in $\Abs_{\lambda_0 \playcircle} G$ (in red), but not in $\Abs_{\lambda_1 \playcircle} G$.
		
		\begin{figure}
			\begin{center}
				\begin{tikzpicture}[->,>=latex,scale=0.7, every node/.style={scale=0.7},initial text={}]
				\node[state, initial left, rectangle, blue] (a) at (0, 0) {$a$};
				
				\node[state, rectangle, orange] (abo) at (2, 2) {$ab^\omega$};
				\node[state, rectangle, orange, dashed] (acdo) at (2, 0) {$acd^\omega$};
				\node[state, rectangle, orange] (aceo) at (2, -2) {$ace^\omega$};
				
				\node[state, rectangle, blue] (ac) at (4, 2) {$ac$};
				\node[state, rectangle, blue] (t) at (5, 0) {$\top$};
				\node[state, rectangle, blue] (ab) at (4, -2) {$ab$};
				
				\node[state, rectangle, blue] (c) at (6, 2) {$c$};
				\node[state, rectangle, blue] (b) at (6, -2) {$b$};
				
				\node[state, rectangle, orange, dashed] (cdo) at (8, 3) {$cd^\omega$};
				\node[state, rectangle, orange] (ceo) at (8, 1) {$ce^\omega$};
				\node[state, rectangle, orange] (bo) at (8, -2) {$b^\omega$};

				\path[red] (a) edge (abo);
				\path[dashed] (a) edge (acdo);
				\path (a) edge (aceo);
				\path (abo) edge (ac);
				\path (abo) edge (t);
				\path[dashed] (acdo) edge (t);
				\path[dashed] (acdo) edge (ab);
				\path (aceo) edge (t);
				\path (aceo) edge (ab);
				\path[red] (ac) edge (c);
				\path[red] (t) edge[loop below] (t);
				\path (ab) edge (b);
				\path[dashed, red] (c) edge (cdo);
				\path (c) edge (ceo);
				\path (b) edge (bo);
				\path[dashed] (cdo) edge (t);
				\path (ceo) edge (t);
				\path (bo) edge (t);
				\end{tikzpicture}
			\end{center}
			\caption{An abstract negotiation game.}
			\label{fig_abstract}
		\end{figure}
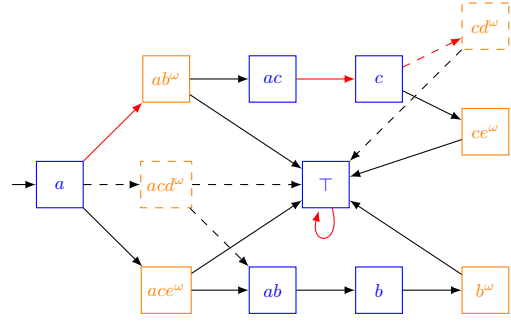
	\end{ex}

	\section{Algorithms} \label{sec_algo}
	
	Let us now study how we can use the abstract negotiation game to solve the problems presented in the introduction.
	We first define an equivalence relation between histories and between plays; then, we show that in the abstract negotiation game, Prover can propose only plays that are simple representatives of their equivalence class, and propose always the same play from each vertex.

	\subsection{Reduced plays and reduced strategy}
	
The equivalence relation that we use is based on the order in which vertices appear.
	
	\begin{defi}[Occurrence-equivalence (histories)]
		Two histories $h$ and $h'$ are \emph{occurrence-equivalent}, written $h \approx h'$, iff $\first(h) = \first(h')$, $\last(h) = \last(h')$ and $\Occ(h) = \Occ(h')$.
	\end{defi}
	
	\begin{defi}[Occurrence-equivalence (plays)]
		Two plays $\rho$ and $\rho'$ are \emph{occurrence-equivalent}, written $\rho \approx \rho'$, iff the three following conditions are satisfied:
		\begin{itemize}
			\item $\Inf(\rho) = \Inf(\rho')$;
			
			\item for each history prefix of $\rho$, there exists a occurrence-equivalent history prefix of $\rho'$;
			
			\item for each history prefix of $\rho'$, there exists a occurrence-equivalent history prefix of $\rho$.
		\end{itemize}
	\end{defi}
	
\begin{ex}
    Let us consider the game of Figure~\ref{fig_ex2}.
    In that game, the play $ab(dedf)^\omega$ is occurrence-equivalent to the play $abde(dedf)^\omega$, but not to the play $ab(dfde)^\omega$.
    Indeed, the latter has the history $abdf$ as a prefix, which is not occurrence-equivalent to any prefix of $ab(dedf)^\omega$, in which the state $f$ occurs only when the state $e$ has already occurred.
\end{ex}
	
	\begin{rk}
		The operators $\Occ$, $\Inf$ and $\mu$ are stable by occurrence-equivalence.
	\end{rk}


	The interest of that equivalence relation lies in the finite number of its equivalence classes, and by the existence of simple representatives of each of them.
	
	\begin{lm}[App.~\ref{pf_lasso}] \label{lm_lasso}
		Let $\rho$ be a play of $G$.
		There exists a lasso $h c^\omega \approx \rho$ with $|h| \leq n^3 + n^2$ and $|c| \leq n^2$, where $n = \card V$.
	\end{lm}

	We call such lassos \emph{reduced plays}.
	For each $\rho$, we write $\trho$ for an arbitrary occurrence-equivalent reduced play.
	Then, operations such as computing $\mu(\trho)$, $\Occ(\trho)$, $\Inf(\trho)$, or checking whether $\trho$ is $\lambda$-consistent, can be done in time $O(n^3)$.

\begin{defi}[Reduced strategy]
    A strategy $\tau_\P$ for Prover in $\Abs_{\lambda i}(G)$ is \emph{reduced} iff it is memoryless, and for each state $v$, the play $\rho$ with $[\rho] = \tau_\P([v])$ is a reduced play.
\end{defi}

\begin{ex}
    In Figure~\ref{fig_abstract}, Prover's winning strategy, defined by the red arrows, is reduced.
\end{ex}

If $\rho \approx \rho'$, and if Challenger can deviate from $\rho$ after the history $hv$, then he can also deviate in $\rho'$ after some history $h'v$ that traverses the same states.
Thus, Prover can play optimally while proposing only reduced plays, and by proposing always the same play from each vertex; that is, by following a reduced strategy.

\begin{lm}[App.~\ref{pf_reduced_strategy}] \label{lm_reduced_strategy}
	Prover has a winning strategy in the abstract negotiation game if and only if she has a reduced one.
\end{lm}

	\subsection{Checking that a reduced strategy is winning: the deviation graph}
	
	We have established that Prover is winning the abstract negotiation game if and only if she has a reduced winning strategy.
	Such a strategy has polynomial size and can thus be guessed in nondeterministic polynomial time.
	It remains us to show that we can verify in deterministic polynomial time that a guessed strategy is winning.
	For that purpose, we construct its \emph{deviation graph}.
	
	\begin{defi}[Deviation graph]
		Let $\tau_\P$ be a reduced strategy of Prover, and let $i \in \Pi$.
		The \emph{deviation graph} associated to $\tau_\P$ and $v$ is the colored graph $\Dev_i(\tau_\P)$:
		\begin{itemize}
			\item the vertices are the plays $\tau_\P([w])$, for every vertex $w$ of the original game;
			
			\item there is an edge from $[\trho]$ to $\tau_\P([w])$ with color $c$ iff there exists $k \in \N$ such that $\trho_k \in V_i$, $\trho_k w \in E$, $w \neq \trho_{k+1}$ and $\min\kappa_i(\Occ(\trho_0 \dots \trho_k)) = c$.
		\end{itemize}
	\end{defi}

	Constructing the deviation graph associated to a memoryless strategy $\tau_\P$ enables to decide whether $\tau_\P$ is a winning strategy or not.
	
	\begin{lm} \label{lm_fixed_point_np_easy}
		The reduced strategy $\tau_\P$ is winning in the abstract negotiation game if and only if in the corresponding deviation graph, there neither exists, from the vertex $\tau_\P([v_0])$:
		\begin{itemize}
			\item a finite path to a vertex $[\trho]$ such that the play $\trho$ is winning for player $i$;
			
			\item nor an infinite path along which the minimal color seen infinitely often is even.
		\end{itemize}
		
		As a consequence, given a parity game $G$ and a requirement $\lambda$, deciding whether $\lambda$ is a fixed point of the negotiation function is $\NP$-easy.
	\end{lm}
	
	\begin{proof}
	    Given $G$ and $\lambda$, let $n$ be the number of states in $G$, and $m$ be the number of colors.
		The deviation graph can be seen as the abstract negotiation game itself, where one removed the transitions that were not compatible with $\tau_\P$; removed the states that were not accessible from $[v_0]$; and merged the paths $[\trho][hv][v][\trho']$ into one edge $[\trho][\trho']$ with color $\min \kappa_i(h)$.
		
		Therefore, a path from the vertex $\tau_\P([v_0])$ can be seen as a history (if it is finite) or a play (if it is infinite), compatible with the strategy $\tau_\P$, in the abstract negotiation game.
		In particular, the finite paths to a vertex $[\trho]$ with $\mu_i(\trho) = 1$ correspond to the histories that lead Prover to propose a play that is winning for player $i$, that Challenger can accept to win.
		Similarly, the infinite paths along which the minimal color seen infinitely often is even correspond to the plays $\pi$ where Challenger deviates infinitely often, and constructs the play $\dpi$ that is winning for player $i$.
		Such paths will be called \emph{winning paths}.
		
		Now, the deviation graph has $n$ vertices, and at most $mn^2$ edges.
		Constructing it requires time $O(n^4)$.
        Deciding the existence of a finite winning path is a simple accessibility problem, and can be done in time $O(mn^2)$.
        Deciding the existence of an infinite winning path is similar to deciding the emptiness of a parity automaton, and requires a time $O(mn^3)$.
		As a consequence, deciding whether a reduced strategy is winning can be done in polynomial time --- and it is an object of polynomial size.
		
		Thus, an $\NP$ algorithm that decides whether the requirement $\lambda$ is a fixed point of the negotiation function is the following: we guess, at the same time, a certificate that proves that $\lambda$ is satisfiable (Lemma~\ref{lm_satisfiable_np_easy}), and a reduced strategy $\tau_\P^v$ for Prover in $\Abs_{\lambda i}(G)_{\|[v]}$, for each $i \in \Pi$ and $v \in V_i$ such that $\lambda(v) = 0$: by Lemma~\ref{lm_reduced_strategy}, there exists such a winning strategy if and only if $\nego(\lambda)(v) = 0$.
		Those objects form a certificate of polynomial size, that can be checked in polynomial time.
\end{proof}

\begin{ex}
    Let us consider the game of Figure~\ref{fig_ex2}, and the requirement $\lambda^*$, presented on the same figure.
    Let $\tau_\P$ be the memoryless strategy in the game $\Abs_{\lambda^*\playcircle} (G)_{\|[a]}$ defined by:
    $$\begin{matrix}
        \tau_\P([a]) = ab(dedf)^\omega & \tau_\P([c]) = c^\omega & \tau_\P([e]) = (edfd)^\omega \\
        \tau_\P([b]) = b(dedf)^\omega &  \tau_\P([d]) = (dedf)^\omega & \tau_\P([f]) = (fded)^\omega.
    \end{matrix}$$
    
    The corresponding deviation graph is given by Figure~\ref{fig_deviation_graph}.
    The purple edges have color $0$, and the orange ones have color $1$.
    Observe that there is no winning path from the vertex $\tau_\P([a]) = ab(dedf)^\omega$: each purple edge can be used at most once, and even though the play $c^\omega$ is winning for player $\Circle$, the corresponding vertex is not accessible.
    Therefore, the strategy $\tau_\P$ is winning in $\Abs_{\lambda^*\playcircle} (G)_{\|[a]}$, which proves the equality $\nego(\lambda^*)(a) = \lambda^*(a) = 0$.
\end{ex}

	\subsection{Upper bounds}
	
	Let us now give the main problems for which the concepts given above yield a solution.
	
\begin{wrapfigure}[24]{r}{0.45\textwidth}\centering
    \begin{tikzpicture}[->,>=latex,scale=0.7, every node/.style={scale=0.7}]
            \node[state, rectangle, rounded corners] (a) at (0:2) {$ab(dedf)^\omega$};
            \node[state, rectangle, rounded corners] (b) at (240:2) {$b(dedf)^\omega$};
            \node[state, rectangle, rounded corners, double] (c) at (120:2) {$c^\omega$};
            \node[state, rectangle, rounded corners] (d) at (60:2) {$(dedf)^\omega$};
            \node[state, rectangle, rounded corners] (e) at (300:2) {$(edfd)^\omega$};
            \node[state, rectangle, rounded corners] (f) at (180:2) {$(fded)^\omega$};
            
            \path[violet] (a) edge (d);
            \path[violet] (a) edge (e);
            \path[violet] (a) edge (f);
            \path[violet, bend left=10] (c) edge (b);
            \path[orange, bend right=10] (d) edge (e);
            \path[orange] (b) edge (d);
            \path[orange] (b) edge (e);
            \path[orange,<->, bend left=10] (f) edge (e);
            \path[orange] (b) edge (f);
            \path[orange, bend left=10] (d) edge (f);
            \path[orange] (d) edge[loop above] (d);
        \path[orange] (e) edge[loop below] (e);
        \path[orange] (f) edge[loop left] (f);
    \end{tikzpicture}
    \caption{A deviation graph.}
    \label{fig_deviation_graph}
    \bigskip
    \begin{tikzpicture}[->,>=latex,shorten >=1pt, scale=0.4, every node/.style={scale=0.6}]
			\node[state] (C1) at (0:5) {$C_1$};
			\node[state] (C11) at (30:6) {$x_1$};
			\node[state] (C12) at (30:4) {$\neg x_1$};
			\node[state] (C2) at (60:5) {$C_2$};
			\node[state] (C21) at (90:6) {$x_2$};
			\node[state] (C22) at (90:4) {$\neg x_2$};
			\node[state] (C3) at (120:5) {$C_3$};
			\node[state] (C31) at (150:6) {$x_3$};
			\node[state] (C32) at (150:4) {$\neg x_3$};
			\node[state] (C4) at (180:5) {$C_4$};
			\node[state] (C41) at (210:6) {$x_4$};
			\node[state] (C42) at (210:4) {$\neg x_4$};
			\node[state] (C5) at (240:5) {$C_5$};
			\node[state] (C51) at (270:6) {$x_5$};
			\node[state] (C52) at (270:4) {$\neg x_5$};
			\node[state] (C6) at (300:5) {$C_6$};
			\node[state] (C61) at (330:6) {$x_6$};
			\node[state] (C62) at (330:4) {$\neg x_6$};
			\node[state] (b) at (0,0) {$\bot$};
			
			\path[->] (C1) edge (C11);
			\path[->] (C1) edge (C12);
			\path[->] (C11) edge (C2);
			\path[->] (C12) edge (C2);
			\path[->] (C11) edge[bend left=30] (b);
			\path[->] (C12) edge (b);
			\path[->] (C2) edge (C21);
			\path[->] (C2) edge (C22);
			\path[->] (C21) edge (C3);
			\path[->] (C22) edge (C3);
			\path[->] (C21) edge[bend left=30] (b);
			\path[->] (C22) edge (b);
			\path[->] (C3) edge (C31);
			\path[->] (C3) edge (C32);
			\path[->] (C31) edge (C4);
			\path[->] (C32) edge (C4);
			\path[->] (C31) edge[bend left=30] (b);
			\path[->] (C32) edge (b);
			\path[->] (C4) edge (C41);
			\path[->] (C4) edge (C42);
			\path[->] (C41) edge (C5);
			\path[->] (C42) edge (C5);
			\path[->] (C41) edge[bend left=30] (b);
			\path[->] (C42) edge (b);
			\path[->] (C5) edge (C51);
			\path[->] (C5) edge (C52);
			\path[->] (C51) edge (C6);
			\path[->] (C52) edge (C6);
			\path[->] (C51) edge[bend left=30] (b);
			\path[->] (C52) edge (b);
			\path[->] (C6) edge (C61);
			\path[->] (C6) edge (C62);
			\path[->] (C61) edge (C1);
			\path[->] (C62) edge (C1);
			\path[->] (C61) edge[bend left=30] (b);
			\path[->] (C62) edge (b);
			\path (b) edge[loop above] (b);
		\end{tikzpicture}
		\caption{The game $G_\phi$.}
		\label{fig_Gphi}
\end{wrapfigure}

	A first application is an algorithm for the SPE constrained existence problem.
	
	\begin{lm} \label{lm_constrained_existence_np_easy}
		The SPE constrained existence problem for parity games is $\NP$-easy.
	\end{lm}
	
	\begin{proof}
		Given a parity game $G_{\|v_0}$ and two thresholds $\bx$ and $\by$, we can guess a reduced play $\teta$ from $v_0$, a requirement $\lambda$, and the certificates required to decide whether $\lambda$ is a fixed point of $\nego$, according to Lemma~\ref{lm_fixed_point_np_easy}.
		All those objects have polynomial size.
		Then, to check that $\teta$ is an SPE outcome, using Theorem~\ref{thm_nego}, we check that $\lambda$ is a fixed point of $\nego$, that $\teta$ is $\lambda$-consistent, and that $\bx \leq \mu(\teta) \leq \by$ in polynomial time.
	\end{proof}

	This algorithm does not need an effective characterization of all the SPEs in a game, which is given by the least fixed point of the negotiation function $\lambda^*$.
	Computing such a characterization can be done with a call to a $\NP$ oracle and to a $\coNP$ oracle, i.e. it belongs to the class $\BH_2$.

	\begin{lm} \label{lm_lfp_bh2_easy}
		Given a parity game $G$ and a requirement $\lambda$, deciding whether $\lambda = \lambda^*$ is $\BH_2$-easy.
	\end{lm}
	
	\begin{proof}
		First, deciding whether $\lambda$ is a fixed point of the negotiation function is an $\NP$-easy problem by Lemma~\ref{lm_fixed_point_np_easy}.
		Deciding whether it is the least one is $\coNP$-easy, because a negative instance can be recognized as follows.
		We guess a requirement $\lambda' < \lambda$, and the certificates of the algorithm given by Lemma~\ref{lm_fixed_point_np_easy}; then, we check that $\lambda'$ is a fixed point of $\nego$.
	\end{proof}

Finally, SPE-verification requires polynomial space.

\begin{lm} \label{lm_spe_verif_pspace_easy}
    The SPE-verification problem is $\PSpace$-easy.
\end{lm}

\begin{proof}
    Given a parity game $G_{\|v_0}$ and an LTL formula $\phi$, by Lemma~\ref{lm_lfp_bh2_easy}, the requirement $\lambda^*$ can be computed by a deterministic algorithm using polynomial space --- indeed, we have the inclusions $\NP \subseteq \PSpace$ and $\coNP \subseteq \PSpace$, hence the guess of $\lambda^*$, followed by one call to an $\NP$ algorithm and one call to a $\coNP$ algorithm, can be transformed into a $\PSpace$ algorithm.
    Then, we can construct in polynomial time the LTL formula $\psi_{\lambda^*}$, that is satisfied exactly by the $\lambda^*$-consistent plays:
    $$\psi_{\lambda^*} = \bigwedge_i \bigwedge_{v \in V_i, \lambda^*(v) = 1} \left( \F v \Rightarrow \bigvee_{2k \leq m} \left( \bigvee_{\kappa_i(w) = 2k} \G \F w \wedge \bigvee_{\kappa_i(w) < 2k} \F \G \neg w \right) \right),$$
    where $m$ is the largest color appearing in $G$.
    
    Then, deciding whether there exists an SPE outcome in $G_{\|v_0}$ that satisfies the formula $\phi$ is equivalent to decide whether there exists a play in $G_{\|v_0}$ that satisfies the formula $\phi \wedge \psi_{\lambda^*}$.
    As for any LTL formula, that can be done using polynomial space: see for example \cite{DBLP:conf/aiml/Schnoebelen02}.
\end{proof}

	\subsection{Fixed-parameter tractability}

We end this section by mentioning an additional complexity result on the SPE constrained existence problem: it is fixed-parameter tractable. 

\begin{thm}[App.~\ref{pf_fpt}] \label{thm_fpt}
    The SPE constrained existence problem on parity games is fixed-parameter tractable when the number of players and the number of colors are parameters.
    More precisely, there exists a deterministic algorithm that solves that problem in time $O(2^{2^{pm}} n^{12})$, where $n$ is the number of vertices, $p$ is the number of players and $m$ is the number of colors.
\end{thm}

\begin{proof}[Proof sketch]
    This result is obtained by constructing and solving a generalized parity game on a finite arena, that is exponential in the number of players and polynomial in the size of the original game.
    This game, called the \emph{concrete} negotiation game, is equivalent to the abstract negotiation game.
    Its construction is inspired by a construction that we have introduced in~\cite{Concur}, for games with mean-payoff objectives.
    The main idea of this construction is to decompose the plays proposed by Prover, by passing them to Challenger edges by edges, and by encoding the $\lambda$-consistence condition into a generalized parity condition.
    Then, we can apply a FPT algorithm to solve generalized parity games that was first proposed in~\cite{DBLP:conf/concur/BruyereHR18}.
    While this deterministic algorithm does not improve on the worst-case complexity of the deterministic $\ExpTime$ algorithm of \cite{DBLP:conf/fsttcs/Ummels06}, it allows for a finer parametric analysis.
\end{proof}

	\section{Matching lower bounds} \label{sec_lower_bounds}
	
	In this section, we provide matching complexity lower bounds for the problems we adressed in the previous section.
	For that purpose, we first need the following construction, inspired from a game designed by Ummels in \cite{GU08} to prove the $\NP$-hardness of the SPE constrained existence problem.
	Our definition is slightly different: while Ummels defined one player per variable, we need one player per literal.
	However, the core intuitions are the same.
	
	\begin{defi}[$G_\phi$]
		Let $\phi = \bigwedge_{j \in \Z/m\Z} C_j$ be a formula of the propositional logic, constructed on the finite set of variables $\{x_1, \dots, x_n\}$.
		We define the parity game $G_\phi$ as follows.
		
		\begin{itemize}
			\item The players are the variables $x_1, \dots, x_n$, their negations, and \emph{Solver}, denoted by $\S$.
			
			\item The states controlled by Solver are all the clauses $C_j$, and the sink state $\bot$.
			
			\item The states controlled by player $L = \pm x_i$ are the pairs $(C_j, L)$, where $L$ is a literal of $C_j$.
			
			\item There are edges from each clause state $C_j$ to all the states $(C_j, L)$; from each pair state $(C_j, L)$ to the state $C_{j+1}$, and to the sink state $\bot$; and from the sink state $\bot$ to itself.
			
			\item For Solver, every state has the color $\kappa_\S(v) = 2$, except the state $\bot$, which has color $1$.
			
			\item For each literal player $L$, every state has the color $\kappa_L(v) = 2$, except the states of the form $(C, \bL)$, that have the color $1$.
		\end{itemize}
	\end{defi}
	
	\begin{rk}
		The game $G_\phi$ is a coBüchi game: Solver has to avoid the sink state $\bot$, and player $L$ the states of the form $(C, \bL)$.
		Therefore, all the following theorems can also be applied to the more restrictive class of coBüchi games.
	\end{rk}
	
	\begin{ex}
		The game $G_\phi$, when $\phi$ is the tautology $\bigwedge_{j=1}^6 (x_j \vee \neg x_j)$, is given by Figure~\ref{fig_Gphi}.
	\end{ex}
	
	The game $G_\phi$ is strongly linked with the satisfiability of $\phi$, in the following formal sense.
	
	\begin{lm}[App.~\ref{pf_Gphi}] \label{lm_Gphi}
		The game $G_\phi$ has the following properties.
		\begin{itemize}
		    \item The least fixed point of the negotiation function is equal to $0$ on the states controlled by Solver, and to $1$ on the other ones.
		    
		    \item For every SPE outcome $\rho$ in $G_\phi$ that does not reach $\bot$, the formula $\phi$ is satisfied by:
		    $$\nu_\rho: x \mapsto \left\{ \begin{matrix}
		        1 & \mathrm{if~} \exists C, (C, x) \in \Inf(\rho) \\
		        0 & \mathrm{otherwise}.
		    \end{matrix} \right.$$
		    
		    \item Conversely, for every valuation $\nu$ satisfying $\phi$, the play $\rho_\nu = (C_1 (C_1, L_1) \dots C_m (C_m, L_m))^\omega,$ where for each $j$, the literal $L_j$ is satisfied by $\nu$, is an SPE outcome.
		\end{itemize}
	\end{lm}

	A first consequence is the lower bound on the complexity of deciding whether some requirement is, or not, a fixed point of the negotiation function.
	
	\begin{thm} \label{thm_fixed_point_np_complete}
		Given a parity game $G$ and a requirement $\lambda$, deciding whether $\lambda$ is a fixed point of the negotiation function is $\NP$-complete.
	\end{thm}
	
	\begin{proof}
		Easiness is given by Lemma~\ref{lm_fixed_point_np_easy}.
		For hardness, we proceed by reduction from the $\NP$-complete problem \textsc{Sat}.
		Given a formula $\phi = \bigwedge_{j \in \Z/m\Z} C_j$ of the propositional logic, we can construct the game $G_\phi$ in polynomial time.
		Then, let us define on $G_\phi$ the requirement $\lambda$ that is constantly equal to $1$, except in $\bot$, where it is equal to $0$.
		Since there is no winning play for Solver from $\bot$, the requirement $\lambda$ is a fixed point of the negotiation function if and only if there exists a $\lambda$-consistent play from each vertex.
		If it is the case, then let $\rho$ be a $\lambda$-consistent play from $C_1$.
		Since $\lambda \geq \lambda^*$, the play $\rho$ is also $\lambda^*$-consistent, and is therefore an SPE outcome.
		It is also a play won by Solver, since $\lambda(C_1) = 1$: therefore, it does not reach the state $\bot$.
		Then, by Lemma~\ref{lm_Gphi}, we can define the valuation $\nu_\rho$, which satisfies $\phi$.
		
		Conversely, if $\nu$ is a valuation satisfying $\phi$, then the play $\rho_\nu$ is an SPE outcome from $C_1$, and it does not end in $\bot$, hence it is won by Solver, and is therefore $\lambda$-consistent.
		If we consider the suffixes of $\rho_\nu$, we find a $\lambda^*$-consistent play from each state $C_j$; and from the states of the form $(C_j, L)$, the play $(C_j, L) \bot^\omega$ is $\lambda^*$-consistent, hence there is a $\lambda^*$-consistent play from each vertex: deciding whether $\lambda$ is a fixed point of the negotiation function is $\NP$-hard.
	\end{proof}

	A similar proof ensures the same lower bound on the SPE constrained existence problem.
	
	\begin{thm} \label{thm_constrained_existence_np_complete}
		The SPE constrained existence problem in parity games is $\NP$-complete.
	\end{thm}
	
	\begin{proof}
		Easiness is given by Lemma~\ref{lm_constrained_existence_np_easy}.
		For hardness, we proceed by reduction from the problem \textsc{Sat}.
		Given a formula $\phi = \bigwedge_{j \in \Z/m\Z} C_j$ of the propositional logic, we can construct the game $G_\phi$ in polynomial time.
		By Lemma~\ref{lm_Gphi}, there exists an SPE outcome from $C_1$ and won by Solver, i.e. an SPE outcome from $C_1$ that does not reach the sink state $\bot$, if and only if there exists a valuation satisfying $\phi$: the SPE constrained existence problem is $\NP$-hard.
	\end{proof}

	Now, we show that the algorithm we presented to compute $\lambda^*$ is optimal as well.
	
	\begin{thm}[App.~\ref{pf_lfp_bh2_complete}] \label{thm_lfp_bh2_complete}
		Given a parity game $G$ and a requirement $\lambda$, deciding whether $\lambda = \lambda^*$ is $\BH_2$-complete.
		Given a parity game $G$, computing $\lambda^*$ can be done by a non-deterministic algorithm in polynomial time iff $\NP = \BH_2$.
	\end{thm}

Finally, the LTL model-checking problem easily reduces to the SPE-verification problem, which gives us the matching lower bound.
	
	\begin{thm}[App.~\ref{pf_spe_verif_pspace_complete}] \label{thm_spe_verif_pspace_complete}
		The SPE-verification problem is $\PSpace$-complete.
	\end{thm}

	\bibliography{Biblio}

\appendix

The following appendices are providing the detailed proofs of some results given in the main body of this paper.
They are not necessary to understand our results and are meant to provide full formalization and rigorous proofs.
To improve readability, we have chosen to recall the statements that appeared in the main body of the paper.

    \section{Proof of Lemma~\ref{lm_sat_cosat}} \label{pf_sat_cosat}
    
\noindent \textbf{Lemma~\ref{lm_sat_cosat}.} \emph{The problem $\Sat \times \coSat$ is $\BH_2$-complete.}
	
	\begin{proof}
		Given a pair $(\phi_1, \phi_2)$, the problem $\Sat \times \coSat$ can be decided by, first, deciding whether $\phi_1$ is satisfiable, which is $\NP$-easy, and second, deciding whether $\phi_2$ is not, which is $\coNP$-easy.
		The problem $\Sat \times \coSat$ is therefore $\BH_2$-easy.
		
		Now, let us prove that it is $\BH_2$-hard; or in other words, that any $\BH_2$-easy problem can be reduced to it in polynomial time.
		Let $P \cap Q$ be a $\BH_2$-easy problem, where $P$ is $\NP$-easy and $Q$ is $\coNP$-easy, and both are defined on the instance set $I$.
		
		Since $\Sat$ is $\NP$-complete, there exists a reduction function $f: I \to \mathcal{F}$, where $\mathcal{F}$ is the set of the instances of $\Sat$, which is computable in polynomial time and such that for every instance $x \in I$, we have $x \in P$ if and only if $f(x)$ is a satisfiable formula.
		Identically, there exists a reduction function $g: I \to \mathcal{F}$ which is computable in polynomial time and such that for every instance $x \in I$, we have $x \in Q$ if and only if $f(x)$ is a non-satisfiable formula.
		
		Then, we can define the function $h: x \mapsto (f(x), g(x))$, which is computable in polynomial time, and which maps every instance $x$ of $P \cap Q$ to an instance $h(x)$ of $\Sat \times \coSat$, such that $h(x)$ is a positive instance of $\Sat \times \coSat$ if and only if $x$ is a positive instance of $P \cap Q$.
		The problem $\Sat \times \coSat$ is $\BH_2$-complete.
	\end{proof}

	\section{Proof of Theorem~\ref{thm_existence_spe}} \label{pf_existence_spe}

\noindent \textbf{Theorem~\ref{thm_existence_spe}.} \emph{There exists an SPE in every parity game.}

	\begin{proof}
		This result has already been proved by Ummels in \cite{DBLP:conf/fsttcs/Ummels06}. Here, we present a variation of his proof, that uses the concepts of requirements and negotiation.
		Let $G$ be a parity game.
		
		Let us define a decreasing sequence $(E_n)_n$ of subsets of $E$, and an associated sequence $(\lambda'_n)_n$ of requirements, keeping the hypothesis that $E_n$ always contains at least one outgoing edge from each vertex.
		First, $E_0 = E$ and $\lambda'_0$ is the vacuous requirement.
		Then, for every $n$ and each player $i$, let us consider the two-player zero-sum game $G_i^n$, defined as the game $G$ where each edge $uv \not\in E_n$ has been removed, and where all the players $j \neq i$ are coalized as one unique player $-i$, whose objective is to make player $i$ lose.
		As proved in \cite{DBLP:conf/fsttcs/Ummels06}, player $i$ has a memoryless uniformly optimal strategy $\sigma^n_i$ in $G_i^n$, i.e. a memoryless strategy that is winning from each state from which player $i$ has a winning strategy.
		We define, for each state $v \in V_i$:
		$$\lambda'_{n+1}(v) = \inf_{\btau_{-i} \in \Sigma_{-i}(G_i^n)} ~\mu_i(\< \btau_{-i}, \sigma^n_i \>_v),$$
		and:
		$$E_{n+1} = E_n \setminus \bigcup_i \{uv ~|~ u \in V_i, \lambda'_{n+1}(u) = 1, \mathrm{~and~} \sigma^n_i(u) \neq v\}.$$
		As desired, the set $E_{n+1}$ contains at least one outgoing edge from each vertex.
		
		Note that as a consequence, all the requirements $\lambda'_n$ have their values in $\{0, 1\}$: we have $\lambda'_{n+1}(v) = 1$ if the strategy $\sigma^n_i$ is winning from $v$, and $0$ otherwise.
		
		We prove by induction that for each $n$, every play that contains only edges from $E_n$ is $\lambda'_n$-consistent.
		It is clear for $n = 0$, since every play is $\lambda_0$-consistent.
		
		Now, if that property is true for some $n \in \N$, let $i \in \Pi$ and $w \in V_i$ be such that $\lambda'_{n+1}(w) \geq 1$: the strategy $\sigma^n_i$ is winning from $w$ in the game $G^n_i$.
		Let then $\rho$ be a play in $G_{\|w}$, that uses only edges from $E_{n+1}$: seen as a play in $G^n_{i\|w}$, it is compatible with the strategy $\sigma^n_i$, and therefore winning for player $i$.
		Every play using only edges from $E_{n+1}$ is $\lambda'_{n+1}$-consistent.
		
		Now, let us prove that for each $n$, we have $\lambda'_{n+1} \geq \nego(\lambda'_n)$.
		Let $i \in \Pi$ and $v \in V_i$ be a vertex such that $\lambda'_{n+1}(v) = 0$: in the game $G^n_i$, the player $-i$ has a strategy $\sigma_{-i}$ against which player $i$ cannot win.
		As a consequence, the strategy $\sigma_{-i}$ prevents player $i$ to access to any state $w \in V_i$ such that $\lambda'_{n+1}(w) = 1$; and therefore, to any state controlled by player $i$ from which edges have been removed.
		In other words, player $i$ could not win against $\sigma_{-i}$ even using edges that are not in $E_n$.
		
		The strategy $\sigma_{-i}$ can be seen as a strategy profile $\bsigma_{-i}$ in the game $G$, which is $\lambda'_n$-rational by the previous result.
		Player $i$ cannot win in $G_{\|v}$ against $\bsigma_{-i}$, hence $\nego(\lambda'_n)(v) = 0$.
		
		But then, since the sequence $(\lambda'_n)_n$ is made of requirements with values in $\{0, 1\}$, it necessarily reaches a limit $\lambda$ with values in $\{0, 1\}$.
		By the previous property, that limit is such that $\lambda \geq \nego(\lambda)$, i.e. it is a fixed point of the negotiation function; and since $\nego(\lambda)(v) < +\infty$ for each $v$, it is a satisfiable requirement.
		Therefore, by Theorem~\ref{thm_nego}, there always exists an SPE in every parity game.
	\end{proof}

    \section{Proof of Lemma~\ref{lm_lasso}} \label{pf_lasso}

\noindent \textbf{Lemma~\ref{lm_lasso}.} \emph{Let $\rho$ be a play of $G$. There exists a lasso $h c^\omega \approx \rho$ with $|h| \leq n^3 + n^2$ and $|c| \leq n^2$, where $n = \card V$.}
	
	\begin{proof}
		Let us write $W_0 \subset \dots \subset W_t$ for all the sets of the form $\Occ(\rho_0 \dots \rho_k)$ with $k \in \N$, without repetition.
		Note that for each index $s \in \{0, \dots, t-1\}$, the set $W_{s+1}$ contains the set $W_s$ plus one additional vertex.
		
		Let us construct the history $h$ and the cycle $c$ as follows, maintaining the hypothesis that for all $p$, the set $\Occ(h_0 \dots h_p)$ is equal to $W_s$ for some $s$.
		
		\begin{itemize}
			\item First, $h_0 = \rho_0$, and $\{h_0\} = W_0$;
			
			\item then, when the prefix $h_0 \dots h_p$ is constructed: let $s$ be such that $\Occ(h_0 \dots h_p) = W_s$, and let $k$ be the minimal integer such that $\rho_0 \dots \rho_k = W_s$.
			
			Let $U$ be the set of all the vertices $u$ such that there exists $\l$ with $\Occ(\rho_0 \dots \rho_\l) = W_s$ and $\rho_\l = u$: any such $\l$ is greater than or equal to $k$, and $U \subseteq W_s$. Then, there exists at least one path from $\rho_k = h_p$ that traverses all the vertices of $U$ and only them. Let $h_p \dots h_q$ be such a path with minimal length: it has at most length $n^2$.
			
			If $s < t$, let now $\l$ be the minimal index greater than $k$ such that $\Occ(\rho_0 \dots \rho_\l) = W_{s+1}$. Then, there exists a path from $h_q \in W_s$ to $\rho_\l$ that uses only vertices of $U$: let $h_q \dots h_r$ be such a path with minimal length. Then, it traverses all vertices at most once, and has therefore length at most $n$.
			
			If $s = t$, let $h_q \dots h_r$ be a path of minimal length from $h_q$ to a vertex $h_r \in \Inf(\rho)$: for the same reasons as above, such a path exists and has length at most $n$. Then, we can stop here the construction of $h$, and observe that the vertex $h_r$ belongs to the graph $(\Inf(\rho), E \cap \Inf(\rho)^2)$, which is strongly connected. We can therefore choose a cycle $c$ that traverses all its vertices and only them, and that has length at most $n^2$.
		\end{itemize}
		
		By construction, the lasso $hc^\omega$ is occurrence-equivalent to $\rho$, and satisfies the desired size conditions.
	\end{proof}

    \section{Proof of Lemma~\ref{lm_reduced_strategy}} \label{pf_reduced_strategy}

\noindent \textbf{Lemma~\ref{lm_reduced_strategy}.} \emph{Prover has a winning strategy in the abstract negotiation game if and only if she has a reduced one.}
	
\begin{proof}
	Let us consider the \emph{reduced negotiation game}, i.e. the abstract negotiation game in which one would have removed all the states but:
	\begin{itemize}
	    \item those of the form $[\trho]$, where $\trho$ is a reduced play;
	    
	    \item those of the form $[hv]$, where $h$ is a prefix of a $\lambda$-consistent reduced play $\trho$, and has minimal length among the occurrence-equivalent prefixes of $\trho$.
	\end{itemize}
	
	By Lemma \ref{lm_lasso}, if $G$ has $n$ vertices, then this game has at most $n^{n^3 + 2n^2} (n^3 + 3n^2) + 1$.
	
	Let us first notice a useful property of that game.
	
	\begin{lm} \label{lm_reduced_game_memoryless}
		Either Prover or Challenger has a memoryless winning strategy in the reduced negotiation game.
	\end{lm}
	
	\begin{proof}
		By \cite{DBLP:conf/icalp/Kopczynski06}, this result is true if both Prover's and Challenger's objectives are prefix-independent and \emph{convex}, on a game with finitely many states.
		By \emph{convex}, we mean that for all winning plays $\pi, \pi'$ and for any two sequences $k_1 < k_2 < \dots$ and $\l_0 < \l_1 < \dots$, with $\pi'_{\l_p} = \pi_{k_q}$ for all $p, q \in \N$, the play:
        $$\chi = \pi_0 \dots \pi_{k_1} \pi'_1 \dots \pi'_{\l_1} \pi_{k_1+1} \dots \pi_{k_2} \pi'_{\l_1+1} \dots$$
        must also be winning.
        The plays of the form of $\chi$ are called \emph{shufflings} of $\pi$ and $\pi'$.
        
        If $\pi$ and $\pi'$ are plays in the reduced negotiation game, we observe that the play $\dchi$ is also a shuffling of the plays $\dpi$ and $\dpi'$.
        The convexity of Prover's and Challenger's objectives is then a consequence of the convexity of parity objectives: the minimal color seen infinitely often by player $i$ in $\dchi$ is the minimum of the minimal colors seen infinitely often in $\dpi$ and in $\dpi'$.
	\end{proof}
	
	We can now prove our theorem by using the equivalence between the abstract and the reduced negotiation game.

		\begin{itemize}
		    \item \emph{If Prover has a winning strategy in the abstract negotiation game, she has one in the reduced negotiation game.}
		    
			Indeed, if Prover has no winning strategy in the reduced game, then, by Lemma \ref{lm_reduced_game_memoryless}, Challenger has a memoryless one: let us write it $\tau_\C$.
			
			Now, let us extend $\tau_\C$ into a memoryless winning strategy $\tau_\C^\star$ in the abstract game.
			
			Let $[\rho] \in S_\C$, and let $[hvw] = \tau_\C([\trho])$.
			By occurrence-equivalence, there exists $k \in \N$ such that $\rho_k = v$, and $\Occ(\rho_0 \dots \rho_k) = \Occ(hv)$.
			We then set $\tau_\C^\star([\rho]) = [\rho_0 \dots \rho_k w]$.
			Note that the states $\tau_\C^\star([\rho]) = [\rho_0 \dots \rho_k w]$ and $\tau_\C([\trho]) = [hvw]$ can be different, but they both have as unique successor the state $[w]$. 
			
			Let us prove, now, that $\tau_\C^\star$ is winning: let $\pi^\star$ be a play compatible with $\tau_\C^\star$.
			When Prover proposes a play $\rho$, in the abstract game, against the strategy $\tau_\C^\star$, and when she proposes the play $\trho$ in the reduced game against the strategy $\tau_\C$, the same thing happens in both cases: either Challenger accepts in both games, or he deviates, and Prover has to propose a new play from the same state $w$.
			Therefore, we can define from $\pi^\star$ a play $\pi$ compatible with $\tau_\C$ in the reduced game, in which each Challenger's state $[\rho]$ is replaced by the state $[\trho]$, and Prover's states are replaced accordingly.
			
			Since the play $\pi$ is compatible with $\tau_\C$, it is winning for Challenger: let us then prove that so is $\pi^\star$.
			If $\pi^\star$ has the form $H [\rho] \top^\omega$, then $\pi$ has the form $H' [\trho] \top^\omega$, hence $\trho$ is winning for player $i$, and therefore $\rho$ is winning for player $i$ and $\pi^\star$ is winning for Challenger.
			If $\pi$ never reaches the state $\top$, then we have $\dpi = h^0 h^1 \dots$ and $\dpi^\star = h^{0\star} h^{1\star} \dots$ where, for every $k$, the histories $h^k$ and $h^{k\star}$ are possibly different, but contain exactly the same vertices.
			Then, the set of player $i$'s colors appearing infinitely often in $\dpi$ and $\dpi^\star$ are the same, and the play $\dpi^\star$ is winning for player $i$, i.e. the play $\pi^\star$ is winning for Challenger: the strategy $\tau^\star_\C$ is winning.
		
			\item \emph{If Prover has a winning strategy in the reduced negotiation game, she has a reduced one in the abstract negotiation game.}
			
			Indeed, let $\tau_\P$ be a winning strategy for Prover in the reduced negotiation game.
			By Lemma \ref{lm_reduced_game_memoryless}, we can assume without loss of generality that $\tau_\P$ is memoryless.
			Since in the abstract game, the only states controlled by Prover where she has several possible choices are the ones of the form $[v]$, for $v \in V$, we can see $\tau_\P$ as a reduced strategy in the abstract game: let us write it $\tau_\P^\star$ in that case.
			We now have to prove that $\tau_\P^\star$ is also a winning strategy.
			
			Let $\pi^\star$ be a play in the abstract negotiation game compatible with $\tau_\P^\star$.
			For any sequence of states $[v] [\trho] [hw] [w]$ that appears in $\pi^\star$, the history $h$ is occurrence-equivalent to some prefix $\th$ of $\trho$ such that $[\th w]$ is a state of the reduced game.
			Therefore, we can transform the play $\pi^\star$ into a play $\pi$ of the reduced game, compatible with $\tau_\P$, where each state of the form $[hw]$ have been replaced by $[\th w]$.
			Since $\tau_\P$ is a winning strategy in the reduced negotiation game, the play $\pi$ is winning for Prover --- let us prove that so is $\pi^\star$.
			
			If $\pi^\star$ has the form $H [\trho] \top^\omega$, then $\pi$ has the form $H' [\trho] \top^\omega$, and since $\pi$ is winning for Prover, the play $\trho$ is losing for player $i$, and therefore the play $\pi^\star$ is winning for Prover.
			If $\pi^\star$ has the form:
			$$\pi^\star = [v_0] [\trho^0] [h^0v_1] [v_1] [\trho^1] [h^1v_2] \dots$$
			then:
			$$\pi = [v_0] [\trho^0] [\th^0v_1] [v_1] [\trho^1] [\th^1v_2] \dots$$
			and since for each $k$, we have $\Occ(h^k) = \Occ(\th^k)$, we find $\Inf(h^0 h^1 \dots) = \Inf(\th^0 \th^1 \dots)$ and therefore, if $\pi$ is winning for Prover, so is $\pi^\star$.
			
			The strategy $\tau_\P^\star$ is a reduced winning strategy.
		\end{itemize}
		
		Therefore, if Prover has a winning strategy in the abstract negotiation game, she has a reduced one.
	\end{proof}

    \section{Proof of Theorem~\ref{thm_fpt}} \label{pf_fpt}

\noindent \textbf{Theorem~\ref{thm_fpt}.} \emph{The SPE constrained existence problem on parity games is fixed-parameter tractable when the number of players and the number of colors are parameters.
    More precisely, there exists a deterministic algorithm that solves that problem in time $O(2^{2^{pm}} n^{12})$, where $n$ is the number of vertices, $p$ is the number of players and $m$ is the number of colors.}

\begin{proof}
    Let $G$ be a parity game, let $i \in \Pi$, and let $v_0 \in V_i$.
    Let us assume, without loss of generality, that $m$ is even, and that the vertices of $G$ or labelled by the colors $0, \dots, m-1$.
    Let $\lambda$ be a requirement.
    The value of $\nego(\lambda)(v_0)$ can be computed using the following lemma.
    
    \begin{lm}
        We have $\nego(\lambda)(v_0) = 0$ if and only if Prover has a winning strategy in the \emph{concrete negotiation game} $\Conc_{\lambda i}(G)_{\|s_0}$, defined as follows:
        \begin{itemize}
            \item the two players are Prover, denoted by $\P$, and Challenger, denoted by $\C$;
            
            \item Prover's states are the states of the form $(u, P)$, of which $u \in V$ is the \emph{current vertex} and $P \subseteq \Pi$ is the \emph{memory};
            
            \item Challenger's states are the states of the form $(uv, P)$, with $uv \in E$ and $P \subseteq \Pi$;
            
            \item there is a transition from each state $(u, P)$ to each state $(uv, P)$ (such a transition is called \emph{proposal}), and from each state $(uv, P)$ to the state $(v, P)$ if $\lambda(v) = 0$ and $(v, P \cup \{j\})$ if $\lambda(v) = 1$, where $v \in V_j$ (\emph{acceptation}), and to each state $(w, \{j\})$, with $w \in V_j$, $w \neq v$ and $uw \in E$ (\emph{deviation});
            
            \item the zero-sum Boolean outcome function $\nu$ is defined as follows: Prover wins a play $\pi$ if and only if the color $\min \hkappa_d(\InfTr(\pi))$ is even for each \emph{dimension} $d \in \Pi \cup \{\star\}$, where $\InfTr(\pi)$ is the set of the transitions used infinitely often in $\pi$, and where the colors $\hkappa_d$ are defined (on transitions) as follows:
            \begin{itemize}
                \item if $d \in \Pi$ (a \emph{non-main dimension}), then $\hkappa_d(st) = m$ when $st$ is a proposal, $\hkappa_d(st) = 0$ when $st$ is a deviation, and:
                $$\hkappa_d((uv, P)(v, Q) = \left\{ \begin{matrix}
                    \kappa_d(v) & \mathrm{if~} d \in P \\
                    m & \mathrm{otherwise},
                \end{matrix} \right.$$
                when $(uv, P)(v, Q)$ is an acceptation;
                
                \item if $d = \star$ (the \emph{main dimension}), then $\hkappa_d(st) = m$ when $st$ is a proposal, and:
                $$\hkappa_d((uv, P), (w, Q))~=~\kappa_i(w) + 1$$
                when $(uv, P)(w, Q)$ is an acceptation or a deviation;
            \end{itemize}
            
            \item the initial state is $s_0 = (v_0, \{i\})$ if $\lambda(v_0) = 1$, and $s_0 = (v_0, \emptyset)$ otherwise.
        \end{itemize}
    \end{lm}
    
    \begin{proof}
        A similar tool has been defined for mean-payoff games in \cite{Concur}.
        As a consequence, this proof is very similar to the proof of Theorem~3 of the same paper.
    
        \begin{itemize}
            \item \emph{If Prover has a winning strategy $\tau_\P$ in the concrete negotiation game.}
            Let $\bsigma$ be the strategy profile in the original game defined as follows.
            For every history $h_0 \dots h_k$ compatible with $\bsigma_{-i}$, let:
            $$H = (h_0, P_0) (h_0h'_0, P_0) \dots (h_k, P_k)$$
            be the only history of that form that is compatible with the strategy $\tau_\P$.
            Let $(h_kv, P_k) = \tau_\P(H)$.
            We define $\bsigma(h_0 \dots h_k) = v$.
            Thus, the strategy profile $\bsigma$ is defined after every history compatible with $\bsigma_{-i}$, and we define it arbitrarily after the other histories.
            
            The strategy profile $\bsigma_{-i}$ is $\lambda$-rational assuming $\sigma_i$: indeed, let $hv = h_0 \dots h_{k-1} v$ be a history compatible with $\bsigma_{-i}$, and let $\rho = \< \bsigma_{\|hv} \>_v$.
            Let us assume that $\rho$ is not $\lambda$-consistent.
            Without loss of generality, we assume that $\mu_j(\rho) = 0$, where $j$ is the player controlling the state $v$, while $\lambda(v) = 1$ (if needed, we extend $hv$ and shorten $\rho$).
            Then, the integer $\min\kappa_j(\Inf(\rho))$ is odd.
            
            Let us now consider, in the concrete negotiation game, the play:
            $$\pi = (h_0, P_0) (h_0h'_0, P_0) \dots (h_{k-1}h'_{k-1}, P_{k-1}) (v, P_k) (v\rho_1, P_k) (\rho_1, P_{k+1}) \dots,$$
            compatible with $\tau_\P$.
            That play contains only proposals and acceptations after the $2k$th step.
            Therefore, the memory $P_k$ is contained in the memories $P_\l$ for $\l > k$, and it contains $j$ since $v \in V_j$ and $\lambda(v) = 1$, hence the transitions $(\rho_\l\rho_{\l+1}, P_{k+\l-1}) (\rho_{\l+1}, P_{k+\l})$ have the colors $\hkappa_j(\rho_\l\rho_{\l+1}, P_{k+\l-1}) = \kappa_j(\rho_\l \rho_{\l+1})$.
            Thus, we have $\min\hkappa_j(\InfTr(\pi)) = \min\kappa_j(\Inf(\rho))$, which is odd: the play $\pi$ is lost by Prover, which is a contradiction since $\tau_\P$ was assumed to be a winning strategy.
            
            Moreover, player $i$ cannot win against the strategy profile $\bsigma_{-i}$.
            Indeed, let $\sigma'_i$ be a strategy for player $i$, and let $\rho = \< \bsigma_{-i}, \sigma'_i \>_{v_0}$.
            Let us consider the play:
            $$\pi = (\rho_0, P_0) (\rho_0\rho'_0, P_0) (\rho_1, P_1) \dots,$$
            compatible with $\tau_\P$.
            Since $\tau_\P$ is a winning strategy, the play $\pi$ is won by Prover, hence the color $\min\hkappa_\star(\InfTr(\pi))$ is even.
            By definition of $\hkappa_\star$, we have $\min\hkappa_\star(\InfTr(\pi)) = \min\kappa_i(\Inf(\rho)) + 1$, hence the color $\min\kappa_i(\Inf(\rho))$ is odd.
            The play $\rho$ is lost by player $i$.
            
            Therefore, there exists a $\lambda$-rational strategy profile from $v_0$ against which player $i$ cannot win, hence $\nego(\lambda)(v_0) = 0$.

            \item \emph{If $\nego(\lambda)(v_0) = 0$.} Let $\bsigma_{-i}$ be a strategy profile, $\lambda$-rational assuming the strategy $\sigma_i$, against which player $i$ cannot win.
            
            Let $\tau_\P$ be the strategy for Prover, in the concrete negotiation game, defined by:
            $$\tau_\P\left((h_0, P_0) (h_0h'_0, P_0) \dots (h_k, P_k)\right) = (h_kv, P_k),$$
            where $v = \bsigma(h_0 \dots h_k)$.
            Let us prove that $\tau_\P$ is a winning strategy.
            
            Let:
            $$\pi = (\rho_0, P_0) (\rho_0\rho'_0, P_0) (\rho_1, P_1) \dots$$
            be a play compatible with $\pi$.
            
            Let us assume that $\pi$ is lost by Prover, i.e. that $\min\hkappa_d(\InfTr(\pi))$ is odd for some dimension $d$.
            
            If $d = \star$, then $\min\hkappa_d(\InfTr(\pi)) = \min\kappa_i(\Inf(\rho)) + 1$, and the color $\min\kappa_i(\Inf(\rho))$ is even.
            Thus, the play $\rho$ is won by player $i$, which is impossible for a play compatible with $\bsigma_{-i}$.
            
            If $d \in \Pi$, then the play $\pi$ contains finitely many deviations (otherwise, we would have $\min\hkappa_d(\InfTr(\pi)) = 0$, even).
            Let $k$ be the least index such that the suffix $(\rho_k, P_k) (\rho_k\rho'_k, P_k) (\rho_{k+1}, P_{k+1}) \dots$ of $\pi$ contains no deviation --- which means that we have $\rho_k\rho_{k+1} \dots = \< \bsigma_{\|\rho_0 \dots \rho_k} \>_{\rho_k}$.
            Then, the sequence $(P_\l)_{\l \geq k}$ is non-decreasing, and for $\l$ great enough, the set $P_\l$ contains $d$ (otherwise, we would have $\min\hkappa_d(\InfTr(\pi))~=~m$, even).
            Therefore, there exists $\l \geq k$ such that $\rho_\l \in V_d$ and $\lambda(\rho_\l) = 1$.
            
            Thus, the colors $\hkappa_d((\rho_{\l'}\rho_{\l'+1}, P_{\l'})(\rho_{\l'+1}, P_{\l'+1}))$, for $\l' \geq \l$, are equal to $\kappa_d(\rho_{\l'}\rho_{\l'+1})$, and the color $\min\hkappa_d(\InfTr(\pi))$, which we assumed to be odd, is equal to the color $\min\kappa_d(\Inf(\rho))$, hence the play $\rho_k\rho_{k+1} \dots$ is lost by player $d$\dots and is therefore not $\lambda$-consistent: contradiction.
            The strategy $\tau_\P$ is a winning strategy.
        \end{itemize}
    \end{proof}
    
    If we add one state in the middle of each transition, i.e. if we decompose each transition $st \in \Delta$ into two transitions $s \delta_{st}$ and $\delta_{st} t$, that game matches with the framework of \cite{DBLP:conf/fossacs/ChatterjeeHP07}.
    That paper provides an algorithm that solves such games in time:
    $$O\left( (\card S)^{2 \card\Omega \card C} \card\Delta \right) \frac{(\card\Omega \card C)!}{(\card C)!^{\card\Omega}},$$
    where $S$ is the set of states, $\Delta$ the set of transitions, $\Omega$ the set of parity conditions and $C$ the set of colors that appear in them.
    In our case, we can therefore solve the concrete negotiation game in time:
    $$\begin{matrix}
        & O\left( (2^p (n + n^2))^{2 (p+1) (m+1)} 2^p n^3 \right) \frac{((p+1) (m+1))!}{(m+1)!^{p+1}} \\
        
        = & n^{O(p^2 m)} \frac{((p+1) (m+1))^{(p+1) (m+1)} e^{(m+1) (p+1)}}{e^{(p+1) (m+1)} (m+1)^{(m+1) (p+1)}} \\
        
        = & n^{O(p^2 m)}.
    \end{matrix}$$
    
    This method requires, therefore, an exponential time, which does not improve the complexity known since \cite{DBLP:conf/fsttcs/Ummels06}.
    
    However, the same game can also be interpreted as a Boolean Büchi game in the sense of \cite{DBLP:conf/concur/BruyereHR18}, i.e. a two-player zero-sum game in which the objective of the first player (Prover) is to validate a Boolean formula whose atoms are Büchi conditions.
    Indeed, Prover's objective can be written:
    $$\bigwedge_{d \in \Pi \cup \{\star\}} \bigvee_{k = 0}^{\frac{m}{2}} \left( \B\{ \delta_{st} ~|~ \kappa_d(st) = 2k\} \wedge \neg \B \{\delta_{st} ~|~ \kappa_d(st) < 2k\} \right),$$
    where $\B(W)$ is the Büchi objective associated to the set $W$, i.e. the objective of visiting infinitely often at least one vertex of $W$.
    That formula has at most $dm$ atoms, and has size $dm$, in a game of size $O(n^2 2^p)$, hence by Proposition~5 from \cite{DBLP:conf/concur/BruyereHR18}, there exists a deterministic algorithm that decides which player has a winning strategy in time:
    $$O\left( 2^{2^{(p+1)m}} (p+1)m + \left( 2^{(p+1)m 2^{(p+1)m}} O(n^2 2^p) \right)^5 \right)
    = 2^{2^{O(pm)}} n^{10}.$$
    
    Therefore, by constructing the concrete negotiation game and applying that algorithm on each vertex, it is possible to compute the requirement $\nego(\lambda)$ in time $2^{2^{O(pm)}} n^{11}$.
    Thus, it is possible to compute the iterations of the negotiation function on $\lambda_0$: since $\nego$ is non-decreasing, its fixed point $\lambda^*$ will be reached in at most $n$ steps, and will therefore be found in time $2^{2^{O(pm)}} n^{12}$.
    
    Once $\lambda^*$ has been computed, given two thresholds $\bx$ and $\by$, the SPE constrained existence problem can be solved by searching, for each tuple $\bc \in \{0, \dots, m-1\}^\Pi$ where $c_i$ whenever $x_i = 1$ and odd whenever $y_i = 0$, a play $\rho$ in $G$ that avoids the set:
    $$W_{\bc} = \{v \in V_i ~|~ c_i \in 2\Z \mathrm{~and~} \lambda^*(v) = 1\},$$
    and such that for each $i$, we have $\min \kappa_i(\Inf(\rho)) = c_i$.
    For a given tuple $\bc$, the existence of a play can be decided by removing all the states of $W_{\bc}$ and the states $v$ such that $\kappa_i(v) < c_i$ for some $i$, then looking for a strongly connected component that contains at least one vertex $v$ with $\kappa_i(v) = c_i$ for each $i$, and finally check whether the vertices of that strongly connected component are accessible from the initial state in $G$ without visiting the states of $W_{\bc}$.
    All those computations can be done in time $O(n)$.
    Thus, once $\lambda^*$ has be computed, the SPE constrained existence problem can be solved in time $O(m^p n)$.
    
    Given a parity game $G_{\|v_0}$ and two thresholds $\bx$ and $\by$, solving the SPE constrained existence problem can be done in time $2^{2^{O(pm)}} n^{12}$, and is therefore fixed parameter tractable with parameters $m$ and $p$.
\end{proof}

    \section{Proof of Lemma~\ref{lm_Gphi}} \label{pf_Gphi}

\noindent \textbf{Lemma~\ref{lm_Gphi}.} \emph{The game $G_\phi$ has the following properties.}
		\begin{itemize}
		    \item \emph{The least fixed point of the negotiation function is equal to $0$ on the states controlled by Solver, and to $1$ on the other ones.}
		    
		    \item \emph{For every SPE outcome $\rho$ in $G_\phi$ that does not reach $\bot$, the formula $\phi$ is satisfied by:}
		    $$\nu_\rho: x \mapsto \left\{ \begin{matrix}
		        1 & \mathrm{if~} \exists C, (C, x) \in \Inf(\rho) \\
		        0 & \mathrm{otherwise}.
		    \end{matrix} \right.$$
		    
		    \item \emph{Conversely, for every valuation $\nu$ satisfying $\phi$, the play $\rho_\nu = (C_1 (C_1, L_1) \dots C_m (C_m, L_m))^\omega,$ where for each $j$, the literal $L_j$ is satisfied by $\nu$, is an SPE outcome.}
		\end{itemize}
	
	\begin{proof}
		\begin{itemize}
			\item If $\lambda_0$ is the requirement constantly equal to $0$, then $\lambda_1 = \nego(\lambda_0)$ is equal to $0$ in every state controlled by Solver (she loses in $\bot$, and if all the other players choose to always go to $\bot$, she loses from everywhere), and to $1$ in every other state (from any state controlled by him, the player $L$ can go to $\bot$, and win).
			This requirement is a fixed point of the negotiation function (the strategy profile against Solver that chooses to always go to $\bot$ is $\lambda_1$-rational), and is therefore the least one.
			
			Therefore, the SPE outcomes are the plays consistent with that requirement, which we will now write $\lambda^*$.
			
			\item \emph{If $\rho$ is an SPE outcome that does not reach the state $\bot$.}
			
			Then, let us prove that the valuation $\nu_\rho$ satisfies $\phi$.
			Since $\rho$ does not end in the sink state $\bot$,
			for each clause $C$ of $\phi$, the state $C$ is visited infinitely often in $\rho$, and therefore so is at least one of its successors, which we will write $(C, L)$.
			
			If $L$ is a positive literal, say $L = x$, then by definition of $\nu_\rho$ we have $\nu_\rho(x) = 1$ and the clause $C$ is satisfied.
			
			If $L$ is a negative literal, say $L = \neg x$, then since $\lambda^*(C, L) = 1$, he must win as a player in $\rho$.
			Therefore, all the states of the form $(C', x)$ are visited finitely often, hence $\nu_\rho(x) = 0$, and the clause $C$ is also satisfied in that case.
			
			The formula $\phi$ is satisfied by the valuation $\nu_\rho$.

			\item \emph{If the valuation $\nu$ satisfies $\phi$.}
			
			Then, the play $\rho_\nu$ is $\lambda^*$-consistent: the only states it traverses on which $\lambda^*$ is equal to $1$ are the ones of the form $(C_j, L_j)$, where $L_j$ is a literal of $C_j$ satisfied by $\nu$.
			But then, that state is controlled by player $L_j$, who wins the play $\rho_\nu$: since the literal $L_j$ is satisfied by $\nu$, the literal $\bL_j$ is not, hence no state of the form $(C, \bL_j)$ is ever traversed.
			
			Therefore, the play $\rho_\nu$ is an SPE outcome.
		\end{itemize}
	\end{proof}

    \section{Proof of Theorem~\ref{thm_lfp_bh2_complete}} \label{pf_lfp_bh2_complete}

\noindent \textbf{Theorem~\ref{thm_lfp_bh2_complete}} \emph{Given a parity game $G$ and a requirement $\lambda$, deciding whether $\lambda = \lambda^*$ is $\BH_2$-complete.
		Given a parity game $G$, computing $\lambda^*$ can be done by a non-deterministic algorithm in polynomial time if and only if $\NP = \BH_2$.}

	\begin{proof}
		The $\BH_2$-easiness of the second problem is given by Lemma~\ref{lm_lfp_bh2_easy}; it of course implies the $\BH_2$-easiness of the first one.
		
		Now, hardness needs to be shown only for the first problem, and it will imply the $\BH_2$-hardness of the second one.
		We proceed by reduction from the problem $\Sat \times \coSat$, which is $\BH_2$-complete by Lemma~\ref{lm_sat_cosat}.
		
		Let $(\phi_1, \phi_2)$ be a pair of formulas.
		We construct in polynomial time a parity game $G$ and a requirement $\lambda$ such that $\lambda$ is the least fixed point of the negotiation function on $G$ if and only if $\phi_1$ is satisfiable and $\phi_2$ is not.
		
		First, we construct the games $G_{\phi_1}$ and $G_{\phi_2}$, supposed to have disjoint state spaces.
		We add to each of them a new state, written $v_1$ in $G_{\phi_1}$ and $v_2$ in $G_{\phi_2}$, and inserted just before the first clause, as shown in Figure~\ref{fig_reduction_lfp1} for $v_1$.
		Note that each of those two vertices has exactly one outgoing edge, so that Challenger has absolutely no choice to make in any play: his defeat of victory will be determined exclusively by the actions of the other players.
		
		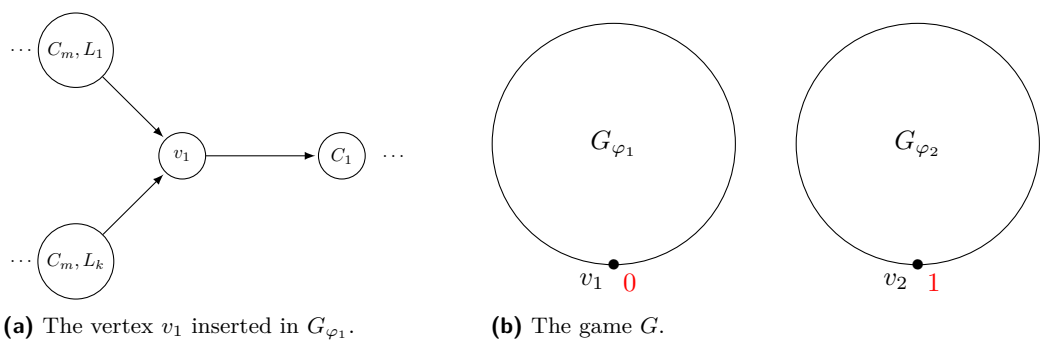
\begin{figure}
			\centering
			\begin{subfigure}[b]{0.45\textwidth}
				\begin{tikzpicture}[->,>=latex,shorten >=1pt, initial text={}, scale=0.7, every node/.style={scale=0.7}]
				\node[state] (a) at (-2, 2) {$C_m, L_1$};
				\node[state] (b) at (-2, -2) {$C_m, L_k$};
				\node[state] (v1) at (0, 0) {$v_1$};
				\node[state] (C1) at (3, 0) {$C_1$};
				\node (a') at (-3, 2) {$\dots$};
				\node (b') at (-3, -2) {$\dots$};
				\node (C1') at (4, 0) {$\dots$};
				
				\path (a) edge (v1);
				\path (b) edge (v1);
				\path (v1) edge (C1);
				\end{tikzpicture}
				\caption{The vertex $v_1$ inserted in $G_{\phi_1}$.}
				\label{fig_reduction_lfp1}
			\end{subfigure}
			\begin{subfigure}[b]{0.45\textwidth}
				\begin{tikzpicture}[scale=0.8, every node/.style={scale=1}]
				\draw (0, 0) circle (2);
				\draw (5, 0) circle (2);
				\draw (0, 0) node {$G_{\phi_1}$};
				\draw (5, 0) node {$G_{\phi_2}$};
				\draw (0, -2) node {$\bullet$};
				\draw (0, -2) node[below left] {$v_1$};
				\draw[red] (0, -2) node[below right] {$0$};
				\draw (5, -2) node {$\bullet$};
				\draw (5, -2) node[below left] {$v_2$};
				\draw[red] (5, -2) node[below right] {$1$};
				\end{tikzpicture}
				\caption{The game $G$.} 
				\label{fig_reduction_lfp2}
			\end{subfigure}
			\caption{Construction of the game $G$.}
		\end{figure}
		
		Those two vertices are controlled by an additional player, \emph{Opponent}, denoted by $\O$, and whose objective is to reach the state $\bot$ --- every state $v$ has the color $\kappa_\O(v) = 1$, except $\bot$, which has the color $\kappa_\O(v) = 2$.
		Note that Opponent's objective is the exact complement of Solver's one.
		The states $v_1$ and $v_2$ are colored to $2$ for every other player than Opponent, so that their objectives stay intuitively the same as in $G_{\phi_1}$ and $G_{\phi_2}$.
		
		The game $G$ is made of those two modified versions of $G_{\phi_1}$ and $G_{\phi_2}$, juxtaposed and not linked in any way, as shows Figure~\ref{fig_reduction_lfp2}.
		The requirement $\lambda$ is equal to $1$ on each state controlled by a literal player, to $0$ on each state controlled by Solver, to $0$ on $v_1$ and to $1$ on $v_2$.
		
		We already know that Solver can never enforce a payoff better than $0$ from her states, and that all the literal players can always enforce the payoff $1$ --- adding $v_1$ and $v_2$ did not change those facts.
		As a consequence, we always have $\lambda^*(v) = \lambda(v)$ for every $v \neq v_1, v_2$.
		
		Now, by Lemma~\ref{lm_Gphi}, there exists an SPE outcome from $v_1$ (resp. $v_2$) that is lost by Opponent if and only the formula $\phi_1$ (resp. $\phi_2$) is satisfiable: therefore, we have $\lambda = \lambda^*$ if and only if the pair $(\phi_1, \phi_2)$ is a positive instance of the problem $\Sat \times \coSat$.
		Deciding whether a given $\lambda$ is the least fixed point of the negotiation function in a given parity game is $\BH_2$-hard, and therefore so is computing that least fixed point.
	\end{proof}

    \section{Proof of Theorem~\ref{thm_spe_verif_pspace_complete}} \label{pf_spe_verif_pspace_complete}

\noindent \textbf{Theorem~\ref{thm_spe_verif_pspace_complete}.} \emph{The SPE-verification problem is $\PSpace$-complete.}
	
	\begin{proof}
		Easiness is given by Lemma~\ref{lm_spe_verif_pspace_easy}.
		
		For hardness, we proceed by reduction from the $\PSpace$-complete problem of model-checking in Kripke structures.
		
		Let $\M = (\A, S, R, (\nu_s)_{s \in S})$ be a Kripke structure, where $\A$ is a finite set of atoms, $S$ is a finite set of states, $R \subseteq S \times S$ is a binary relation over $S$, and each $\nu_s$ is a valuation over $\A$.
		Let $\phi$ be an LTL formula over $\A$.
		
		We associate to $\M$ the game $G = (\Pi, V, (V_i)_{i \in \Pi}, E, \mu)$ where $\Pi = \{\S\}$, $V = V_\S = S$, $E = R$, and $\mu$ is constantly equal to $1$ --- the game $G$ is a parity game with all colors equal to $2$.
		Then, every play in $G$ is an SPE outcome.
		To $\phi$, we associate the LTL formula $\psi$ over the atom set $V$, by replacing each atom $a$ in $\phi$ by the proposition:
		$$\bigvee_{\nu_v(a) = 1} v.$$
		
		Then, there exists a path in $\M$ satisfying $\phi$ if and only if there exists an SPE outcome in $G$ satisfying $\psi$, which yields the desired lower bound.
	\end{proof}

    \section{An example} \label{app_ex}

Let $G$ be the parity game of Figure~\ref{fig_ex3}.
Let us compute the iterations of the negotiation function, i.e. the sequence $(\lambda_n)_n = (\nego^n(\lambda_0))_n$, on that game.

At first, the requirement $\lambda_1$ is equal to $1$ on each state that is controlled by a player who has a winning strategy from that state: namely, the states $b$, $f$, $j$, and $k$.
It is equal to $0$ everywhere else.

At the second step, the requirement $\lambda_2$ is also equal to $1$ on the state $i$, because from there, player $\Circle$ can go to the state $j$, from which only a play winning for player $\Box$ can be proposed to her --- that is, only the play $jk^\omega$, which is winning for her too.
For the same reason, the requirement $\lambda_2$ is also equal to $1$ on the state $g$, because the only $\lambda_1$-consistent losing play that could be proposed to player $\Circle$ from there is $gehi^\omega$, which also enables her to deviate and go to the state $j$.
The requirement $\lambda_2$ is still equal to $0$ on $a$, because the play $ab(cd)^\omega$, winning for player $\Circle$, can be proposed to player $\Box$, without letting him the possibility of deviating.
It is also still equal to $0$ on $c$, $d$ and $h$, because from $e$, the play $(efgeh)^\omega$, winning for player $\Circle$, can be proposed to player $\Diamond$ --- and the play $i^\omega$ if she deviates and goes to $i$.

At the third step, though, if she deviates and goes to $i$, since $\lambda_2(i) = 1$, then the only $\lambda_2$-consistent play that can be proposed to her is $ijk^\omega$, winning for her.
As a consequence, the requirement $\lambda_3$ is equal to $1$ on the states $c, d$, and $h$.

Then, from the state $a$, the play $ab(cd)^\omega$ is no longer $\lambda_3$-consistent, but the play $abcd(ef)^\omega$ still is, as well as the play $(ef)^\omega$ from $e$, hence $\lambda_3 = \lambda^*$.

We sum up those results in the following table.

\begin{center}
    \begin{tabular}{|c|c|c|c|c|c|c|c|c|c|c|c|}
        \hline
        & $a$ & $b$ & $c$ & $d$ & $e$ & $f$ & $g$ & $h$ & $i$ & $j$ & $k$ \\
        \hline
        $\lambda_1$ & $0$ & $1$ & $0$ & $0$ & $0$ & $1$ & $0$ & $0$ & $0$ & $1$ & $1$ \\
        \hline
        $\lambda_2$ & $0$ & $1$ & $0$ & $0$ & $0$ & $1$ & $1$ & $0$ & $1$ & $1$ & $1$ \\
        \hline
        $\lambda_3 = \lambda^*$ & $0$ & $1$ & $1$ & $1$ & $0$ & $1$ & $1$ & $1$ & $1$ & $1$ & $1$ \\
        \hline
    \end{tabular}
\end{center}

\begin{figure}
	\begin{center}
		\begin{tikzpicture}[->,>=latex,scale=0.8, every node/.style={scale=0.8},initial text={}]
		    \node[state, rectangle] (a) at (0, 0) {$a$};
			\node[state] (b) at (2, 0) {$b$};
			\node[state, diamond] (c) at (4, 0) {$c$};
		    \node[state, diamond] (d) at (6, 0) {$d$};
			\node[state, rectangle] (e) at (8, 0) {$e$};
			\node[state] (f) at (9.4, -2) {$f$};
			\node[state] (g) at (6.6, -2) {$g$};
			\node[state, diamond] (h) at (10, 0) {$h$};
			\node[state] (i) at (12, 0) {$i$};
			\node[state, rectangle] (j) at (14, 0) {$j$};
			\node[state, diamond] (k) at (16, 0) {$k$};
				
			\path[<->] (a) edge (b);
			\path[<->] (b) edge (c);
			\path[<->] (c) edge (d);
			\path (d) edge (e);
			\path[<->] (e) edge (f);
			\path (f) edge (g);
			\path (g) edge (e);
			\path[<->] (e) edge (h);
			\path (h) edge (i);
			\path (i) edge[loop below] (i);
			\path (i) edge (j);
			\path (j) edge[loop below] (j);
			\path (j) edge (k);
			\path (k) edge[loop below] (k);
				
			\node (a') at (0, 0.8) {$\stackrel{\playcircle}{3} \stackrel{\Box}{1} \stackrel{\Diamond}{1}$};
			\node (b') at (2, 0.8) {$\stackrel{\playcircle}{2} \stackrel{\Box}{0} \stackrel{\Diamond}{1}$};
			\node (c') at (4, 0.8) {$\stackrel{\playcircle}{1} \stackrel{\Box}{1} \stackrel{\Diamond}{1}$};
			\node (d') at (6, 0.8) {$\stackrel{\playcircle}{0} \stackrel{\Box}{2} \stackrel{\Diamond}{1}$};
			\node (e') at (8, 0.8) {$\stackrel{\playcircle}{2} \stackrel{\Box}{1} \stackrel{\Diamond}{2}$};
			\node (f') at (10.2, -2) {$\stackrel{\playcircle}{2} \stackrel{\Box}{1} \stackrel{\Diamond}{2}$};
			\node (g') at (5.8, -2) {$\stackrel{\playcircle}{1} \stackrel{\Box}{1} \stackrel{\Diamond}{1}$};
			\node (h') at (10, 0.8) {$\stackrel{\playcircle}{0} \stackrel{\Box}{1} \stackrel{\Diamond}{3}$};
			\node (i') at (12, 0.8) {$\stackrel{\playcircle}{1} \stackrel{\Box}{1} \stackrel{\Diamond}{1}$};
			\node (j') at (14, 0.8) {$\stackrel{\playcircle}{1} \stackrel{\Box}{1} \stackrel{\Diamond}{1}$};
			\node (k') at (16, 0.8) {$\stackrel{\playcircle}{0} \stackrel{\Box}{0} \stackrel{\Diamond}{0}$};
		\end{tikzpicture}
	\end{center}
	\caption{A parity game.}
	\label{fig_ex3}
\end{figure}

\end{document}